\documentclass[orivec,conference]{llncs}
\usepackage[cmex10]{amsmath}
\usepackage{stmaryrd}
\usepackage{amsfonts}
\usepackage{amssymb}
\usepackage{bbold}
\usepackage{wrapfig}
\usepackage{graphicx}
\usepackage{enumerate}
\usepackage{multirow}
\usepackage{xspace}
\usepackage{cancel}
\usepackage{enumitem}
\usepackage[hidelinks]{hyperref}
\usepackage{esvect}
\usepackage{caption}
\usepackage{xcolor,colortbl}

\newcommand{\ict}{\ti{IC3}\xspace}

\newcommand{\data}{\mbox{$\mi{Data}$}\xspace}
\newcommand{\fifo}{\mbox{$\mi{Fifo}$}\xspace}

\newcommand{\cv}{\mbox{$\mi{CAV02}$-$\mi{QE}$}\xspace}
\newcommand{\Eg}{\mbox{$\mi{EG}$-$\mi{PQE}$}\xspace}
\newcommand{\egp}{\mbox{$\mi{EG}$-$\mi{PQE}^+$}\xspace}

\newcommand{\dpqe}{$\mi{DS}$-$\mi{PQE}$\xspace}

\newcommand{\cad}{\mbox{$\mi{CADET}$}\xspace}

\newcommand{\olds}[4]{\mbox{$(\prob{#1}{#2},\pnt{#3}\,)\rightarrow #4$}}

\newcommand{\bm}[1]{{\mbox{\boldmath $#1$}}}
\newcommand{\Bm}[1]{{\boldmath $#1$}}

\newcommand{\di}[1]{\mbox{$\mi{Diam}(#1)$}\xspace}

\newcommand{\pl}{\mbox{$\mi{Plg}$}\xspace}
\newcommand{\spe}{\mbox{\Sub{P}{\!agg}}\xspace}
\newcommand{\sspe}{\mbox{$P^*_{\mi{agg}}$}\xspace}

\newcommand{\imp}{\Rightarrow} 
\newcommand{\ie}{i.e.,\xspace}
\newcommand{\eg}{e.g.,\xspace}
\newcommand{\pnt}[1]{\mbox{$\vv{#1}$}\xspace}
\newcommand{\ppnt}[2]{\mbox{$\vv{#1}_{#2}$}}
\newcommand{\pqnt}[2]{\mbox{$\vv{#1}\!_{\mi{#2}}$}}
\newcommand{\pent}[2]{\mbox{$\vv{#1}\!_{#2}$}}
\newcommand{\Cof}[2]{\mbox{$#1^{\vec{#2}}$}}
\newcommand{\cof}[2]{\mbox{$#1_{\vec{#2}}$}}

\newcommand{\V}[1]{\mbox{$\mathit{Vars}(#1)$}}
\newcommand{\Va}[1]{\mbox{$\mi{Vars}(\vec{#1})$}}

\newcommand{\s}[1]{\mbox{$\{#1\}$}}

\newcommand{\nGz}[2]{$G_{non-\{z\}}$}

\newcommand{\prr}[1]{\mi{Prev}(\boldsymbol{q})}

\newcommand{\nf}[1]{\normalfont{#1}}
\newcommand{\mi}[1]{\mathit{#1}}
\newcommand{\ti}[1]{\textit{#1}}
\newcommand{\tb}[1]{\textbf{#1}}

\newcommand{\ttt}{\>\>\>}

\newcommand{\Tt}{\>\>}
\newcommand{\Sub}[2]{\mbox{$\mi{#1}_{\mi{#2}}$}}

\newcommand{\prob}[2]{\mbox{$\exists{#1} [#2]$}}
\newcommand{\Prob}[3]{\mbox{$\exists{#1}\exists{#2}[#3]$}}

\newcommand{\Comment}[1]{}

\newcommand{\Abs}[1]{\mbox{$X_{#1}$}}

\begin{document}

\title{Partial Quantifier Elimination And Property Generation}


\author{Eugene Goldberg}
\institute{\email{eu.goldberg@gmail.com}}

\maketitle

\begin{abstract}
We study partial quantifier elimination (PQE) for propositional CNF
formulas with existential quantifiers. PQE is a generalization of
quantifier elimination where one can limit the set of clauses taken
out of the scope of quantifiers to a small subset of clauses. The
appeal of PQE is that many verification problems (\eg equivalence
checking and model checking) can be solved in terms of PQE and the
latter can be dramatically simpler than full quantifier
elimination. We show that PQE can be used for property generation that
can be viewed as a generalization of testing. The objective here is to
produce an \ti{unwanted} property of a design implementation thus
exposing a bug. We introduce two PQE solvers called \Eg and \egp. \Eg
is a very simple SAT-based algorithm. \egp is more sophisticated and
robust than \Eg.  We use these PQE solvers to find an unwanted
property (namely, an unwanted invariant) of a buggy FIFO buffer. We
also apply them to invariant generation for sequential circuits from a
HWMCC benchmark set. Finally, we use these solvers to generate
properties of a combinational circuit that mimic symbolic simulation.
\end{abstract}

\vspace{-10pt}
\section{Introduction}
\label{sec:intro}
In this paper, we consider the following problem. Let $F(X,Y)$ be a
propositional formula in conjunctive normal form
(CNF)\footnote{Every formula is a propositional CNF formula unless otherwise
stated. Given a CNF formula $F$ represented as the conjunction of
clauses $C_1 \wedge \dots \wedge C_k$, we will also consider $F$ as
the \ti{set} of clauses \s{C_1,\dots,C_k}.
} where $X,Y$ are sets of
variables. Let $G$ be a subset of clauses of $F$.  Given a
formula \prob{X}{F}, find a quantifier-free formula $H(Y)$ such that
$\prob{X}{F}\equiv H\wedge\prob{X}{F \setminus G}$.  In contrast
to \ti{full} quantifier elimination (QE), only the clauses of $G$ are
taken out of the scope of quantifiers here. So, we call this problem 
\ti{partial} QE (PQE)~\cite{hvc-14}.  (In this paper, we
consider PQE only for formulas with \ti{existential} quantifiers.) We
will refer to $H$ as a \ti{solution} to PQE.  Like SAT, PQE is a way
to cope with the complexity of QE. But in contrast to SAT that is a
\ti{special} case of QE (where all variables are quantified), PQE
\ti{generalizes} QE. The latter is just a special case of PQE where $G = F$
and the entire formula is
unquantified. Interpolation~\cite{craig,ken03} is a special case of
PQE as well~\cite{tech_rep_pc_lor} (see also
Appendix~\ref{app:interp}).

The appeal of PQE is threefold. First, it can be much more efficient
than QE if $G$ is a \ti{small} subset of $F$. Second, many
verification problems like SAT, equivalence checking, model checking
can be solved in terms of
PQE~\cite{hvc-14,south_korea,fmcad16,mc_no_inv2}. So, PQE can be used
to design new efficient methods for solving known problems. Third, one
can apply PQE to solving \ti{new} problems like property generation
considered in this paper. In practice, to perform PQE, it suffices to
have an algorithm that takes a single clause out of the scope of
quantifiers. Namely, given a formula \prob{X}{F(X,Y)} and a clause
$C \in F$, this algorithm finds a formula $H(Y)$ such that
$\prob{X}{F} \equiv H \wedge \prob{X}{F \setminus \s{C}}$. To take out
$k$ clauses, one can apply this algorithm $k$ times.  Since
$H \wedge \prob{X}{F} \equiv H\wedge \prob{X}{ F \setminus \s{C}}$,
solving the PQE above reduces to finding $H(Y)$ that makes
$C$ \ti{redundant} in $H \wedge \prob{X}{F}$. So, the PQE algorithms
we present here employ \ti{redundancy based reasoning}.  Namely, we
describe two PQE algorithms called \Eg and \egp where ``\ti{EG}''
stands for ``Enumerate and Generalize''. \Eg is a very simple
SAT-based algorithm that can sometimes solve very large problems. \egp
is a modification of \Eg that makes the algorithm more powerful and
robust.

In~\cite{fmcad16}, we showed the viability of an equivalence checker
based on PQE.  In particular, we presented instances for which this
equivalence checker outperformed ABC~\cite{abc}, a high quality
tool. In this paper, we describe and check experimentally one more
important application of PQE called property generation.  Our
motivation here is as follows. Suppose a design
implementation \ti{Imp} meets the set of specification properties
$P_1,\dots,P_m$.  Typically, this set is incomplete. So, \ti{Imp} can
still be buggy even if every $P_i,i=1,\dots,m$ holds.  Let
$P^*_{m+1},\dots,P^*_n$ be \ti{desired} properties adding which makes
the specification complete. If \ti{Imp} meets the properties
$P_1,\dots,P_m$ but is still buggy, a missed property $P^*_i$ above
fails. That is, \ti{Imp} has the \ti{unwanted} property
$\overline{P^*_i}$. So, one can detect bugs by generating unspecified
properties of \ti{Imp} and checking if there is an unwanted one.

Currently, identification of unwanted properties is mostly done by
massive testing. (As we show later, the input/output behavior
specified by a single test can be cast as a simple property
of \ti{Imp}.) Another technique employed in practice is \ti{guessing}
a list of unwanted properties that may hold and formally checking if
this is the case.  The problem with these techniques is that they can
miss an unwanted property. In this paper, we describe property
generation by PQE.  The benefit of PQE is that it can produce much
more complex properties than those corresponding to single tests. So,
using PQE one can detect bugs that testing overlooks or cannot find in
principle. Importantly, PQE generates properties covering different
parts of \ti{Imp}. This makes the search for unwanted properties more
systematic and facilitates discovering bugs that can be missed if one
simply guesses unwanted properties that may hold.

In this paper, we experimentally study generation of invariants of a
sequential circuit $N$. An invariant of $N$ is unwanted if a state
that is supposed to be reachable in $N$ falsifies this invariant and
hence is unreachable. Note that finding a formal proof that $N$ has no
unwanted invariants is impractical. (It is hard to efficiently prove a
large set of states reachable because different states are reached by
different execution traces.)  So developing practical methods for
finding unwanted invariants if very important. We also study
generation of properties mimicking symbolic simulation for a
combinational circuit obtained by unrolling a sequential circuit. An
unwanted property here exposes a wrong execution trace.

The main body of this paper is structured as follows. (Some additional
information can be found in appendices.) In Section~\ref{sec:basic},
we give basic definitions. Section~\ref{sec:prop_gen} presents
property generation for a combinational circuit. In
Section~\ref{sec:inv_gen}, we describe invariant generation for a
sequential circuit. Sections~\ref{sec:eg_pqe} and~\ref{sec:eg_pqe+}
present \Eg and \egp respectively.  In Section~\ref{sec:fifo_exper},
invariant generation is used to find a bug in a FIFO buffer.
Experiments with invariant generation for HWMCC benchmarks are
described in Section~\ref{sec:inv_gen_exper}.
Section~\ref{sec:comb_exper} presents an experiment with property
generation for combinational circuits.  In Section~\ref{sec:bg} we
give some background. Finally, in Section~\ref{sec:concl}, we make
conclusions and discuss directions for future research.

\vspace{-2pt}
\section{Basic Definitions}
\vspace{-1pt}
\label{sec:basic}

In this section, when we say ``formula'' without mentioning
quantifiers, we mean ``a quantifier-free formula''.

\begin{definition}
\label{def:cnf}
We assume that formulas have only Boolean variables.  A \tb{literal}
of a variable $v$ is either $v$ or its negation.  A \tb{clause} is a
disjunction of literals. A formula $F$ is in conjunctive normal form
(\tb{CNF}) if $F = C_1 \wedge \dots \wedge C_k$ where $C_1,\dots,C_k$
are clauses. We will also view $F$ as the \tb{set of
clauses} \s{C_1,\dots,C_k}. We assume that \tb{every formula is in
CNF}.
\end{definition}

%
%
\begin{definition}
  \label{def:vars} Let $F$ be a formula. Then \bm{\V{F}} denotes the
set of variables of $F$ and \bm{\V{\prob{X}{F}}} denotes
$\V{F}\!\setminus\!X$.
\end{definition}

%
%
\begin{definition}
Let $V$ be a set of variables. An \tb{assignment} \pnt{q} to $V$ is a
mapping $V'~\rightarrow \s{0,1}$ where $V' \subseteq V$.  We will
denote the set of variables assigned in \pnt{q}~~as \bm{\Va{q}}. We will
refer to \pnt{q} as a \tb{full assignment} to $V$ if $\Va{q}=V$. We
will denote as \bm{\pnt{q} \subseteq \pnt{r}} the fact that a) $\Va{q}
\subseteq \Va{r}$ and b) every variable of \Va{q} has the same value
in \pnt{q} and \pnt{r}.
\end{definition}

%
%
\begin{definition}
A literal, a clause, and a formula are said to be \tb{satisfied}
(respectively \tb{falsified}) by an assignment \pnt{q} if they
evaluate to 1 (respectively 0) under \pnt{q}.
\end{definition}

%
%
\begin{definition}
\label{def:cofactor}
Let $C$ be a clause. Let $H$ be a formula that may have quantifiers,
and \pnt{q} be an assignment to
\V{H}.  If $C$ is satisfied by \pnt{q}, then \bm{\cof{C}{q} \equiv
  1}. Otherwise, \bm{\cof{C}{q}} is the clause obtained from $C$ by
removing all literals falsified by \pnt{q}. Denote by \bm{\cof{H}{q}}
the formula obtained from $H$ by removing the clauses satisfied by
\pnt{q} and replacing every clause $C$ unsatisfied by \pnt{q} with
\cof{C}{q}.
\end{definition}

%
%
\begin{definition}
  \label{def:Xcls}
Given a formula \prob{X}{F(X,Y)}, a clause $C$ of $F$ is called a
\tb{quantified clause} if \V{C} $\cap~X~\neq~\emptyset$. If $\V{C}
\cap X = \emptyset$, the clause $C$ depends only on free, \ie
unquantified variables of\, $F$ and is called a \tb{free clause}.
\end{definition}

%
%
\begin{definition}
\label{def:formula-equiv}
Let $G, H$ be formulas that may have existential quantifiers. We say
that $G, H$ are \tb{equivalent}, written \bm{G \equiv H}, if $\cof{G}{q} =
\cof{H}{q}$ for all full assignments \pnt{q} to $\V{G} \cup \V{H}$.
\end{definition}

%
%
\begin{definition}
\label{def:red_cls}
Let $F(X,Y)$ be a formula and $G \subseteq F$ and $G \neq
\emptyset$. The clauses of $G$ are said to be \textbf{redundant in} \bm{\prob{X}{F}} if
$\prob{X}{F} \equiv \prob{X}{F \setminus G}$. Note that if $F
\setminus G$ implies $G$, the clauses of $G$ are redundant in
\prob{X}{F}.
\end{definition}

%
%
\begin{definition}
 \label{def:pqe_prob} Given a formula \prob{X}{F(X,Y))} and $G$ where
 $G \subseteq F$, the \tb{Partial Quantifier Elimination} (\tb{PQE})
 problem is to find $H(Y)$ such that\linebreak \Bm{\prob{X}{F}\equiv
 H\wedge\prob{X}{F \setminus G}}.  (So, PQE takes $G$ out of the scope
 of quantifiers.)  The formula $H$ is called a \tb{solution} to
 PQE. The case of PQE where $G = F$ is called \tb{Quantifier
 Elimination} (\tb{QE}).
\end{definition}

%
%
\begin{example}
\label{exmp:pqe_exmp}
Consider the formula $F = C_1 \wedge C_2 \wedge C_3 \wedge C_4$ where
$C_1=\overline{x}_3 \vee x_4$, $C_2\!=\!y_1\!\vee\!x_3$,
$C_3=y_1 \vee \overline{x}_4$, $C_4\!=\!y_2\!\vee\!x_4$. Let $Y$
denote \s{y_1,y_2} and $X$ denote \s{x_3,x_4}. Consider the PQE
problem of taking $C_1$ out of \prob{X}{F}, \ie finding $H(Y)$ such
that $\prob{X}{F} \equiv H \wedge \prob{X}{F \setminus \s{C_1}}$. As
we show later, $\prob{X}{F} \equiv
y_1 \wedge \prob{X}{F \setminus \s{C_1}}$.  That is, $H\!  =\!y_1$ is
a solution to the PQE problem above.
\end{example}

%
%
\begin{remark}
\label{rem:noise}
Let $D$ be a clause of a solution $H$ to the PQE problem of
Definition~\ref{def:pqe_prob}.  If $F \setminus G$ implies $D$, then
$H \setminus \s{D}$ is a solution to this PQE problem too.
\end{remark}

%
%
\begin{proposition}
\label{prop:sol_impl}
Let $H$ be a solution to the PQE problem of
Definition~\ref{def:pqe_prob}.  That is, $\prob{X}{F}\equiv
H\wedge\prob{X}{F \setminus G}$. Then $F \imp H$ (\ie $F$ implies
$H$).
\end{proposition}

The proofs of propositions are given in Appendix~\ref{app:proofs}.

%
%
\begin{definition}
\label{def:resol}
Let clauses $C'$,$C''$ have opposite literals of exactly one variable
$w\!\in\!\V{C'}\!\cap\!\V{C''}$.  Then $C'$,$C''$ are called
\tb{resolvable} on~$w$.  The clause $C$ having all literals of
$C',C''$ but those of $w$ is called the \tb{resolvent} of
$C'$,$C''$. The clause $C$ is said to be obtained by \tb{resolution}
on $w$.
\end{definition}

%
\begin{definition}
  \label{def:blk_cls} Let $C$ be a clause of a formula $G$ and
$w \in \V{C}$. The clause $C$ is said to
be \tb{blocked}~\cite{blocked_clause} in $G$ with respect to the
variable $w$ if no clause of $G$ is resolvable with $C$ on $w$.
\end{definition}
%
%
\begin{proposition}
\label{prop:blk_cls}
Let a clause $C$ be blocked in a formula $F(X,Y)$ with respect to a
variable $x \in X$.  Then $C$ is redundant in \prob{X}{F},
\ie \prob{X}{F \setminus \s{C}} $\equiv$ \prob{X}{F}.
\end{proposition}

\section{Property Generation By PQE}
\label{sec:prop_gen}
Many known problems can be formulated in terms of PQE, thus
facilitating the design of new efficient algorithms.  In
Appendix~\ref{app:using_pqe}, we recall some results on solving SAT,
equivalence checking and model checking by PQE presented in
~\cite{hvc-14,south_korea,fmcad16,mc_no_inv2}.  In this section, we
describe application of PQE to \ti{property generation} for a
combinational circuit.  The objective of property generation is to
expose a bug via producing an \ti{unwanted} property.

Let $M(X,V,W)$ be a combinational circuit where $X,V,W$ specify the
sets of the internal, input, and output variables of $M$
respectively. Let $F(X,V,W)$ denote a formula specifying $M$. As
usual, this formula is obtained by Tseitin's
transformations~\cite{tseitin}. Namely, $F$ equals $F_{G_1} \wedge
\dots \wedge F_{G_k}$ where $G_1,\dots,G_k$ are the gates of $M$ and
$F_{G_i}$ specifies the functionality of gate $G_i$.
%
%
\begin{example}
\label{exmp:gate_cnf}
Let $G$ be a 2-input AND gate defined as $x_3 = x_1 \wedge x_2$ where
$x_3$ denotes the output value and $x_1,x_2$ denote the input values
of $G$. Then $G$ is specified by the formula $F_G\!=\!(\overline{x}_1
\vee \overline{x}_2\vee x_3) \wedge (x_1 \vee \overline{x}_3) \wedge
(x_2 \vee \overline{x}_3)$. Every clause of $F_G$ is falsified by an
inconsistent assignment (where the output value of $G$ is not implied
by its input values). For instance, $x_1\!\vee \overline{x}_3$ is
falsified by the inconsistent assignment $x_1\!=\!0, x_3\!=\!1$. So,
every assignment \ti{satisfying} $F_G$ corresponds to a
\ti{consistent} assignment to $G$ and vice versa. Similarly, every
assignment satisfying the formula $F$ above is a consistent assignment
to the gates of $M$ and vice versa.
\end{example}

\vspace{-5pt}
\subsection{High-level view of property generation by PQE}
One generates properties by PQE until an unwanted property exposing a
bug is produced.  (Like in testing one runs tests until a bug-exposing
test is encountered.)  The benefit of property generation by PQE is
fourfold.  First, by property generation one can identify bugs that
are hard or simply impossible to find by testing. Second, using PQE
makes property generation efficient.  Third, by taking out different
clauses one can generate properties covering different parts of the
design. This increases the probability of discovering a bug.  Fourth,
every property generated by PQE specifies a large set of high-quality
tests.

In this paper (Sections~\ref{sec:fifo_exper},~\ref{sec:comb_exper}),
we consider cases where identifying an unwanted property is
easy. However, in general, such identification is not trivial.  A
detailed discussion of this topic is beyond the scope of this paper.
(Nevertheless, in Appendix~\ref{app:unw_props}, we give the main idea
of the procedure for deciding if a property is unwanted.)

%
%
\subsection{Property generation as generalization of testing}
\label{ssec:pg_gnrl_tsts}
The behavior of $M$ corresponding to a single test can be cast as a
property.  Let $w_i \in W$ be an output variable of $M$ and \pnt{v} be
a test, \ie a full assignment to the input variables $V$ of $M$. Let
\Cof{B}{v} denote the longest clause falsified by \pnt{v}, \ie
$\V{\Cof{B}{v}} = V$. Let $l(w_i)$ be the literal satisfied by the
value of $w_i$ produced by $M$ under input \pnt{v}. Then the clause
$\Cof{B}{v} \vee l(w_i)$ is satisfied by every assignment satisfying
$F$, \ie $\Cof{B}{v} \vee l(w_i)$ is a property of $M$.  We will refer
to it as a \tb{single-test property} (since it describes the behavior
of $M$ for a single test). If the input \pnt{v} is supposed to produce
the opposite value of $w_i$ (\ie the one \ti{falsifying} $l(w_i)$),
then \pnt{v} exposes a bug in $M$.  In this case, the single-test
property above is an \tb{unwanted} property of $M$ exposing the same
bug as the test \pnt{v}.

A single-test property can be viewed as a weakest property of $M$ as
opposed to the strongest property specified by \prob{X}{F}.  The
latter is the truth table of $M$ that can be computed explicitly by
performing QE on \prob{X}{F}.  One can use PQE to generate properties
of $M$ that, in terms of strength, range from the weakest ones to the
strongest property inclusively. (By combining clause splitting with
PQE one can generate single-test properties, see the next subsection.)
Consider the PQE problem of taking a clause $C$ out of \prob{X}{F}.
Let $H(V,W) $ be a solution to this problem, \ie $\prob{X}{F} \equiv H
\wedge \prob{X}{F \setminus \s{C}}$.  Since $H$ is implied by $F$, it
can be viewed as a \tb{property} of $M$. If $H$ is an \tb{unwanted}
property, $M$ has a bug. (Here we consider the case where a property
of $M$ is obtained by taking a clause out of formula \prob{X}{F} where
only the \ti{internal} variables of $M$ are quantified. Later we
consider cases where some external variables of $M$ are quantified
too.)

We will assume that the property $H$ generated by PQE has no redundant
clauses (see Remark~\ref{rem:noise}). That is, if $D \in H$, then $F
\setminus \s{C} \not\imp D$. Then one can view $H$ as a property that
holds due to the presence of the clause $C$ in $F$.

%
%
\subsection{Computing properties efficiently}
\label{ssec:eff_comp}
If a property $H$ is obtained by taking only one clause out of
\prob{X}{F}, its computation is much easier than performing QE on
\prob{X}{F}. If computing $H$ still remains too time-consuming, one
can use the two methods below that achieve better performance at the
expense of generating weaker properties. The first method applies when
a PQE solver forms a solution \ti{incrementally}, clause by clause
(like the algorithms described in Sections~\ref{sec:eg_pqe}
and~\ref{sec:eg_pqe+}). Then one can simply stop computing $H$ as soon
as the number of clauses in $H$ exceeds a threshold.  Such a formula
$H$ is still implied by $F$ and hence specifies a property of $M$.

The second method employs \ti{clause splitting}.  Here we consider
clause splitting on input variables $v_1,\dots,v_p$, \ie those of $V$
(but one can split a clause on any subset of variables from
\V{F}). Let $F'$ denote the formula $F$ where a clause $C$ is replaced
with $p+1$ clauses: $C_1 = C \vee \overline{l(v_1)}$,\dots, $C_p = C
\vee \overline{l(v_p)}$, $C_{p+1} = C \vee l(v_1) \vee \dots \vee
l(v_p)$, where $l(v_i)$ is a literal of $v_i$. The idea is to obtain a
property $H$ by taking the clause $C_{p+1}$ out of \prob{X}{F'} rather
than $C$ out of \prob{X}{F}. The former PQE problem is simpler than
the latter since it produces a weaker property $H$. One can show that
if $\s{v_1,\dots,v_p}\!  =\!V$, then a) the complexity of PQE reduces
to \tb{linear}; b) taking out $C_{p+1}$ actually produces a
\tb{single-test property}. The latter specifies the input/output
behavior of $M$ for the test \pnt{v} falsifying the literals
$l(v_1),\dots, l(v_p)$. (See Appendix~\ref{app:cls_split} for more
details.)

%
%
\subsection{Using design coverage for generation of unwanted properties}
\label{ssec:des_cov}
Arguably, testing is so effective in practice because one verifies a
\ti{particular design}. Namely, one probes different parts of this
design using some coverage metric rather than sampling the truth table
(which would mean verifying \ti{every possible design}). The same idea
works for property generation by PQE for the following two
reasons. First, by taking out a clause, PQE generates a property
inherent to the \ti{specific} circuit $M$. (If one replaces $M$ with
an equivalent but structurally different circuit, PQE will generate
different properties.)  Second, by taking out different clauses of $F$
one generates properties corresponding to different parts of $M$ thus
``covering'' the design. This increases the chance to take out a
clause corresponding to the buggy part of $M$ and generate an unwanted
property.

\vspace{-3pt}
\subsection{High-quality tests specified by a property generated by PQE}
\label{ssec:tests_props}
In this subsection, we show that a property $H$ generated by PQE, in
general, specifies a large set of high-quality tests. Let $H(V,W)$ be
obtained by taking $C$ out of \prob{X}{F(X,V,W)}. Let $Q(V,W)$ be a
clause of $H$. As mentioned above, we assume that $F \setminus \s{C}
\not\imp Q$. Then there is an \smallskip assignment
(\pnt{x},\pnt{v},\pnt{w}) satisfying formula $(F \setminus \s{C})
\wedge \overline{Q}$ where \pnt{x},\pnt{v},\pnt{w} are assignments to
$X,V,W$ respectively. (Note that by definition, (\pnt{v},\pnt{w})
falsifies $Q$.) Let $(\pnt{x}^*,\pnt{v},\pnt{w}^*)$ be the execution
trace of $M$ under the input \pnt{v}. So,
$(\pnt{x}^*,\pnt{v},\pnt{w}^*)$ satisfies $F$. Note that the output
assignments \pnt{w} and $\pnt{w}^*$ must be different because
$(\pnt{v},\pnt{w}^*)$ has to satisfy $Q$. (Otherwise,
$(\pnt{x}^*,\pnt{v},\pnt{w}^*)$ satisfies $F \wedge \overline{Q}$ and
so $F \not\imp Q$ and hence $F \not\imp H$.)  So, one can view \pnt{v}
as a test ``detecting'' disappearance of the clause $C$ from $F$.
Note that different assignments satisfying $(F \setminus \s{C}) \wedge
\overline{Q}$ correspond to different tests \pnt{v}. So, the clause
$Q$ of $H$, in general, specifies a very large number of tests. One
can show that these tests are similar to those detecting stuck-at
faults and so have very high quality (see
Appendix~\ref{app:tests_props} for more details).

\section{Invariant Generation By PQE}
\label{sec:inv_gen}
In this section, we extend property generation for combinational
circuits to sequential ones. Namely, we generate \ti{invariants}. Note
that generation of \ti{desired} auxiliary invariants is routinely used
in practice to facilitate verification of a predefined property. The
problem we consider here is different in that our goal is to produce
an \ti{unwanted} invariant. We picked generation of invariants (over
that of weaker properties just claiming that a state cannot be reached
in $k$ transitions or less) because identification of an unwanted
invariant is, arguably, easier. This simplifies bug detection by
property generation.

%
%
\subsection{Bugs making states unreachable}
\label{ssec:unr_bugs}
Let $N$ be a sequential circuit and $S$ denote the state variables of
$N$. Let $I(S)$ specify the initial state \pqnt{s}{ini} (\ie
$I(\pqnt{s}{ini})\!=\!1$). Let $T(S',V,S'')$ denote the transition
relation of $N$ where $S',S''$ are the present and next state
variables and $V$ specifies the (combinational) input variables. We
will say that a state \pnt{s} of $N$ is reachable if there is an
execution trace leading to \pnt{s}. That is, there is a sequence of
states $\pnt{s_0},\dots,\pnt{s_k}$ where $\pnt{s_0} = \pqnt{s}{ini}$,
$\pent{s}{k}\!=\!\pnt{s}$ and there exist \pent{v}{i}
$i=0,\dots,k\!-\!1$ for which $T(\pent{s}{i},\pent{v}{i},
\pent{s}{i+1}) = 1$.  Let $N$ have to satisfy a set of \tb{invariants}
$P_0(S),\dots,P_m(S)$. That is, $P_i$ holds iff $P_i(\pnt{s}) = 1$ for
every reachable state \pnt{s} of $N$. We will denote the \tb{aggregate
  invariant} $P_0 \wedge \dots \wedge P_m$ as \bm{\Sub{P}{agg}}. We
will call \pnt{s} a \tb{bad state} of $N$ if
$\Sub{P}{agg}(\pnt{s})=0$. If \Sub{P}{agg} holds, no bad state is
reachable. We will call \pnt{s} a \tb{good state} of $N$ if
$\Sub{P}{agg}(\pnt{s}) = 1$.

Typically, the set of invariants $P_0,\dots,P_m$ is incomplete in the
sense that it does not specify all states that must be
\ti{unreachable}. So, a good state can well be unreachable. We will
call a good state \tb{operative} (or \tb{op-state} for short) if it is
supposed to be used by $N$ and so should be \ti{reachable}. We
introduce the term \ti{an operative state} just to factor out
``useless'' good states. We will say that $N$ has an \tb{op-state
  reachability bug} if an op-state is unreachable in $N$. In
Section~\ref{sec:fifo_exper}, we consider such a bug in a FIFO
buffer. The fact that \Sub{P}{agg} holds says \ti{nothing} about
reachability of op-states. Consider, for instance, a trivial circuit
\Sub{N}{triv} that simply stays in the initial state \pqnt{s}{ini} and
$\Sub{P}{agg}(\pqnt{s}{ini})=1$. Then \Sub{P}{agg} holds for
\Sub{N}{triv} but the latter has op-state reachability bugs (assuming
that the correct circuit must reach states other than \pqnt{s}{ini}).

Let $R_{\vv{s}}(S)$ be the predicate satisfied only by a state
\pnt{s}. In terms of CTL, identifying an op-state reachability bug
means finding \pnt{s} for which the property $EF.R_{\vv{s}}$ must hold
but it does not.  The reason for assuming \pnt{s} to be \ti{unknown}
is that the set of op-states is typically too large to \ti{explicitly}
specify every property $ET.R_{\vv{s}}$ to hold. This makes finding
op-state reachability bugs very hard. The problem is exacerbated by
the fact that reachability of different states is established by
\ti{different traces}. So, in general, one cannot efficiently prove
many properties $EF.R_{\vv{s}}$ (for different states) \ti{at once}.

%
%
\subsection{Proving op-state unreachability  by invariant generation}
\label{ssec:check_unreach}
In practice, there are two methods to check reachability of operative
states for large circuits.  The first method is testing. Of course,
testing cannot prove a state unreachable, however, the examination of
execution traces may point to a potential problem. (For instance,
after examining execution traces of the circuit \Sub{N}{triv} above
one realizes that many operative states look unreachable.) The other
method is to check \tb{unwanted invariants}, \ie those that are
supposed to fail. If an unwanted invariant holds for a circuit, the
latter has an op-state reachability bug. For instance, one can check
if a state variable $s_i \in S$ of a circuit never changes its initial
value. To break this unwanted invariant, one needs to find an
operative state where the initial value of $s_i$ is flipped. (For the
circuit \Sub{N}{triv} above this unwanted invariant holds for every
state variable.) The potential unwanted invariants are formed
manually, \ie simply \ti{guessed}.

The two methods above can easily overlook an op-state reachability
bug.  Testing cannot prove that an op-state is unreachable. To
correctly guess an unwanted invariant that holds, one essentially has
to know the underlying bug. Below, we describe a method for invariant
generation by PQE that is based on property generation for
combinational circuits.  The appeal of this method is twofold.  First,
PQE generates invariants ``inherent'' to the implementation at hand,
which drastically reduces the set of invariants to explore. Second,
PQE is able to generate invariants related to different parts of the
circuit (including the buggy one).  This increases the probability of
generating an unwanted invariant. We substantiate this intuition in
Section~\ref{sec:fifo_exper}.

Let formula \bm{F_k} specify the combinational circuit obtained by
unfolding a sequential circuit $N$ for $k$ time frames and adding the
initial state constraint $I(S_0)$. That is, $F_k = I(S_0) \wedge
T(S_0,V_0,S_1) \wedge \dots \wedge T(S_{k-1},V_{k-1},S_k)$ where $S_j,
V_j$ denote the state and input variables of $j$-th time frame
respectively. Let $H(S_k)$ be a solution to the PQE problem of taking
a clause $C$ out of \prob{\Abs{k}}{F_k} where $\Abs{k} = S_0 \cup V_0
\cup \dots \cup S_{k-1} \cup V_{k-1}$.  That is, \prob{\Abs{k}}{F_k}
$\equiv H \wedge$ \prob{\Abs{k}}{F_k \setminus \s{C}}. Note that in
contrast to Section~\ref{sec:prop_gen}, here some external variables
of the combinational circuit (namely, the input variables
$V_0,\dots,V_{k-1}$) are quantified too. So, $H$ depends only on state
variables of the last time frame. $H$ can be viewed as a \tb{local
  invariant} asserting that no state falsifying $H$ can be reached in
$k$ transitions.

One can use $H$ to find global invariants (holding for \ti{every} time
frame) as follows.  Even if $H$ is only a local invariant, a clause
$Q$ of $H$ can be a \ti{global} invariant.  The experiments of
Section~\ref{sec:inv_gen_exper} show that, in general, this is true
for many clauses of $H$. (To find out if $Q$ is a global invariant,
one can simply run a model checker to see if the property $Q$ holds.)
Note that by taking out different clauses of $F_k$ one can produce
global single-clause invariants $Q$ relating to different parts of
$N$. From now on, when we say ``an invariant'' without a qualifier we
mean a \tb{global invariant}.

\section{Introducing \Eg}
\label{sec:eg_pqe}

In this section, we describe a simple SAT-based algorithm for
performing PQE called \Eg. Here \ti{'EG'} stands for 'Enumerate and
Generalize'.  \Eg accepts a formula \prob{X}{F(X,Y)} and a clause $C
\in F$. It outputs a formula $H(Y)$ such that $\prob{X}{\Sub{F}{ini}}
\equiv H \wedge \prob{X}{\Sub{F}{ini} \setminus \s{C}}$ where
\Sub{F}{ini} is the initial formula $F$. (This point needs
clarification because \Eg changes $F$ by adding clauses.)

%
%
\subsection{An example}
\label{ssec:exmp}
Before describing the pseudocode of \Eg, we explain how it solves the
PQE problem of Example~\ref{exmp:pqe_exmp}.  That is, we consider
taking clause $C_1$ out of \prob{X}{F(X,Y)} where $F = C_1 \wedge
\dots \wedge C_4$, $C_1=\overline{x}_3 \vee x_4$,
$C_2\!=\!y_1\!\vee\!x_3$, $C_3=y_1 \vee \overline{x}_4$,
$C_4\!=\!y_2\!\vee\!x_4$ and $Y=\s{y_1,y_2}$ and $X=\s{x_3,x_4}$.

\Eg iteratively generates a full assignment \pnt{y} to $Y$ and checks
if \cof{(C_1)}{y} is redundant in \prob{X}{\cof{F}{y}} (\ie if $C_1$
is redundant in \prob{X}{F} in subspace \pnt{y}). Note that if \cof{(F
  \setminus \s{C_1})}{y} \ti{implies} \cof{(C_1)}{y}, then
\cof{(C_1)}{y} is trivially redundant in \prob{X}{\cof{F}{y}}. To
avoid such subspaces, \Eg generates \pnt{y} by searching for an
assignment (\pnt{y},\pnt{x}) satisfying the formula $(F \setminus
\s{C_1}) \wedge \overline{C}_1$. (Here \pnt{y} and \pnt{x} are full
assignments to $Y$ and $X$ respectively.) If such (\pnt{y},\pnt{x})
exists, it satisfies $F \setminus \s{C_1}$ and falsifies $C_1$ thus
proving that \cof{(F \setminus \s{C_1})}{y} \ti{does not} imply
\cof{(C_1)}{y}.

Assume that \Eg found an assignment\smallskip
$(y_1\!=\!0,y_2\!=\!1,x_3\!=\!1,x_4\!=\!0)$ satisfying $(F \setminus
\s{C_1}) \wedge \overline{C}_1$. So \pnt{y} =
$(y_1\!=\!0,y_2\!=\!1)$. Then \Eg checks if \cof{F}{y} is satisfiable.
\cof{F}{y} = $(\overline{x}_3 \vee x_4) \wedge x_3 \wedge
\overline{x}_4$ and so it is \ti{unsatisfiable}.  This means that
\cof{(C_1)}{y} \ti{is not} redundant in \prob{X}{\cof{F}{y}}. (Indeed,
\cof{(F \setminus \s{C_1})}{y} is satisfiable. So, removing $C_1$
makes $F$ satisfiable in subspace \pnt{y}.) \Eg \ti{makes}
\cof{(C_1)}{y} redundant in \prob{X}{\cof{F}{y}} by \tb{adding} to $F$
a clause $B$ falsified by \pnt{y}. The clause $B$ equals $y_1$ and is
obtained by identifying the assignments to individual variables of $Y$
that made \cof{F}{y} unsatisfiable. (In our case, this is the
assignment $y_1 = 0$.)  Note that derivation of clause $y_1$
\ti{generalizes} the proof of unsatisfiability of $F$ in subspace
$(y_1\!=\!0,y_2\!=\!1)$ so that this proof holds for subspace
$(y_1\!=\!0,y_2\!=\!0)$ too.\smallskip

Now \Eg looks for a new assignment satisfying $(F \setminus \s{C_1})
\wedge \overline{C}_1$. Let the assignment $(y_1
=1,y_2=1,x_3=1,x_4=0)$ be found. So, \pnt{y} =
$(y_1\!=\!1,y_2\!=\!1)$. Since $(y_1\!=\!1,y_2\!=\!1,x_3=0)$ satisfies
$F$, the formula \cof{F}{y} is satisfiable. So, \cof{(C_1)}{y} is
\ti{already redundant} in \prob{X}{\cof{F}{y}}. To avoid re-visiting
the subspace \pnt{y}, \Eg generates the \tb{plugging} clause $D =
\overline{y}_1 \vee \overline{y}_2$ falsified by \pnt{y}.

\Eg fails to generate a new assignment \pnt{y} because the
formula\linebreak \mbox{$D \wedge (F \setminus \s{C_1}) \wedge
  \overline{C}_1$} is unsatisfiable. Indeed, every full assignment
\pnt{y} we have examined so far falsifies either the clause $y_1$
added to $F$ or the plugging clause $D$. The only assignment \Eg has
not explored \smallskip yet is $\pnt{y}\!=\!(y_1\!=\!1, y_2\!=\!
0)$. Since \mbox{$\cof{(F \setminus \s{C_1})}{y} = x_4$} and
\cof{(C_1)}{y} = $\overline{x}_3 \vee x_4$, the formula $(F \setminus
\s{C_1}) \wedge \overline{C}_1$ is unsatisfiable in subspace
\pnt{y}. In other words,\smallskip \cof{(C_1)}{y} is implied by
\cof{(F \setminus \s{C_1})}{y} and hence is redundant. Thus, $C_1$ is
redundant in \prob{X}{\Sub{F}{ini} \wedge y_1} for every assignment to
$Y$ where \Sub{F}{ini} is the initial formula $F$. That is,
\prob{X}{\Sub{F}{ini}} $\equiv y_1 \wedge$ \prob{X}{\Sub{F}{ini}
  \setminus \s{C_1}} and so the clause $y_1$ is a solution $H$ to
our PQE problem.

%
%
\subsection{Description of \Eg}
%
%
\setlength{\intextsep}{4pt}
\setlength{\textfloatsep}{4pt}
\begin{wrapfigure}{l}{2in}
\centering
\small
\vspace{-5pt}
\parbox{0cm}{\begin{tabbing}
aaa\=b\=cc\= dd\= \kill
$\Eg(F,X,Y,C)$ \{ \\
\scriptsize{1}\> $\pl := \emptyset$; $\Sub{F}{ini}:= F$  \\
\scriptsize{2}\> while (\ti{true}) \{ \\
\scriptsize{3}\Tt $G\!:=F \setminus \s{C}$ \\
\scriptsize{4}\Tt $\pnt{y}\!:=\!\mi{Sat}_1(\pl\!\wedge G\!\wedge\!\overline{C})$  \\
\scriptsize{5}\Tt if ($\pnt{y} = \mi{nil}$) \\
\scriptsize{6}\ttt return($F \setminus \Sub{F}{ini}$) \\
\scriptsize{7}\Tt $(\pnt{x}^*,B) := \mi{Sat}_2(F,\pnt{y})$\\
\scriptsize{8}\Tt if ($B \neq \mi{nil}$) \{  \\
\scriptsize{9}\ttt  $ F:= F \cup \s{B}$ \\
\scriptsize{10}\ttt  continue \}\\
\scriptsize{11}\Tt $D \!:=\!PlugCls(\pnt{y},\!\pnt{x}^*,\!F)$ \\
\scriptsize{12}\Tt $\pl := \pl \cup \s{D}$\}\} \\
\end{tabbing}}
\vspace{-15pt}
\caption{Pseudocode of \Eg}
\label{fig:eg_pqe}
\end{wrapfigure}

The pseudo-code of \Eg is shown in Fig.~\ref{fig:eg_pqe}.  \Eg starts
with storing the initial formula $F$ and initializing formula \pl that
accumulates the plugging clauses generated by \Eg (line 1). As we
mentioned in the previous subsection, plugging clauses are used to
avoid re-visiting the subspaces where the formula $F$ is proved
satisfiable.

All the work is carried out in a while loop. First, \Eg checks if
there is a new subspace \pnt{y} where \prob{X}{\cof{(F \setminus
    \s{C})}{y}} does not imply \cof{F}{y}. This is done by searching
for an assignment (\pnt{y},\pnt{x}) satisfying $\pl \wedge (F
\setminus \s{C}) \wedge \overline{C}$ (lines 3-4). If such an
assignment does not exist, the clause $C$ is redundant in
\prob{X}{F}. (Indeed, let \pnt{y} be a full assignment to $Y$. The
formula $\pl \wedge (F \setminus \s{C}) \wedge \overline{C}$ is
unsatisfiable in subspace \pnt{y} for one of the two reasons. First,
\pnt{y} falsifies \pl. Then \cof{C}{y} is redundant because \cof{F}{y}
is satisfiable. Second, $\cof{(F \setminus \s{C})}{y} \wedge
\overline{\cof{C}{y}}$ is unsatisfiable.  In this case, \cof{(F
  \setminus \s{C})}{y} implies \cof{C}{y}.)  Then \Eg returns the set
of clauses added to the initial formula $F$ as a solution $H$ to the
PQE problem (lines 5-6).

If the satisfying assignment (\pnt{y},\pnt{x}) above exists, \Eg
checks if the formula \cof{F}{y} is satisfiable (line 7).  If not,
then the clause \cof{C}{y} \ti{is not} redundant in
\prob{X}{\cof{F}{y}} (because \cof{(F \setminus \s{C})}{y} is
satisfiable).  So, \Eg \ti{makes} \cof{C}{y} redundant by generating a
clause $B(Y)$ falsified by \pnt{y} and adding it to $F$ (line 9).
Note that adding $B$ also prevents \Eg from re-visiting the subspace
\pnt{y} again.  The clause $B$ is built by finding an
\ti{unsatisfiable} subset of \cof{F}{y} and collecting the literals of
$Y$ removed from clauses of this subset when obtaining \cof{F}{y} from
$F$.

If \cof{F}{y} is satisfiable, \Eg generates an assignment $\pnt{x}^*$
to $X$ such that $(\pnt{y},\pnt{x}^*)$ satisfies $F$ (line 7). The
satisfiability of \cof{F}{y} means that every clause of \cof{F}{y}
including \cof{C}{y} is redundant in \prob{X}{\cof{F}{y}}. At this
point, \Eg uses the longest clause $D(Y)$ falsified by \pnt{y} as a
plugging clause (line 11). The clause $D$ is added to \pl to avoid
re-visiting subspace \pnt{y}. Sometimes it is possible to remove
variables from \pnt{y} to produce a shorter assignment $\pnt{y}^*$
such that $(\pnt{y}^*,\pnt{x}^*)$ still satisfies $F$. Then one can
use a shorter plugging clause $D$ involving only the variables
assigned in $\pnt{y}^*$.
%
%
\subsection{Discussion}
\label{ssec:disc1}
\Eg is similar to the QE algorithm presented at
CAV-2002~\cite{blocking_clause}. We will refer to it as \cv. Given a
formula \prob{X}{F(X,Y)}, \cv enumerates full assignments to $Y$. In
subspace \pnt{y}, if \cof{F}{y} is unsatisfiable, \cv adds to $F$ a
clause falsified by \pnt{y}.  Otherwise, \cv generates a plugging
clause $D$. (In~\cite{blocking_clause}, $D$ is called ``a blocking
clause''. This term can be confused with the term ``blocked clause''
applied to a completely different kind of a clause. So, we use the
term ``a plugging clause'' instead.) To apply the idea of \cv to PQE,
we reformulated it in terms of redundancy based reasoning.

The main flaw of \Eg inherited from \cv is the necessity to use
plugging clauses produced from a satisfying assignment. Consider the
PQE problem of taking a clause $C$ out of \prob{X}{F(X,Y)}.  If $F$ is
proved \ti{unsatisfiable} in subspace \pnt{y}, typically, only a small
subset of clauses of \cof{F}{y} is involved in the proof. Then the
clause generated by \Eg is short and thus proves $C$ redundant in many
subspaces different from \pnt{y}. On the contrary, to prove $F$
\ti{satisfiable} in subspace \pnt{y}, every clause of $F$ must be
satisfied. So, the plugging clause built off a satisfying assignment
includes almost every variable of $Y$. Despite this flaw of \Eg, we
present it for two reasons. First, it is a very simple SAT-based
algorithm that can be easily implemented. Second, \Eg has a powerful
advantage over \cv since it solves PQE rather than QE. Namely, \Eg
does not need to examine the subspaces \pnt{y} where $C$ is implied by
$F \setminus \s{C}$. Surprisingly, for many formulas this allows \Eg
to \ti{completely avoid} examining subspaces where $F$ is satisfiable.
In this case, \Eg is very efficient and can solve very large
problems. Note that when \cv performs complete QE on \prob{X}{F}, it
\ti{cannot} avoid subspaces \pnt{y} where \cof{F}{y} is satisfiable
unless $F$ \ti{itself} is unsatisfiable (which is very rare in
practical applications).

\section{Introducing \egp}
\label{sec:eg_pqe+}
In this section, we describe \egp, an improved version of \Eg.  
%
%
\vspace{-7pt}
\subsection{Main idea}
The pseudocode of \egp is shown in Fig~\ref{fig:eg_pqe+}. It is
different from that of \Eg only in line 11 marked with an asterisk.
The motivation for this change is as follows.  Line 11 describes
proving redundancy of $C$ for the case where \cof{C}{y} is not implied
by \cof{(F \setminus \s{C})}{y} and \cof{F}{y} is satisfiable. Then
\Eg simply uses a satisfying assignment as a proof of redundancy of
$C$ in subspace \pnt{y}. This proof is unnecessarily strong because it
proves that \ti{every} clause of $F$ (including $C$) is redundant in
\prob{X}{F} in subspace \pnt{y}. Such a strong proof is hard to
generalize to other subspaces.

%
%
\setlength{\intextsep}{4pt}
\setlength{\textfloatsep}{4pt}
\begin{wrapfigure}{l}{2in}
\centering
\small
\vspace{-5pt}
\parbox{0cm}{\begin{tabbing}
aaa\=b\=cc\= dd\= \kill
$\egp(F,X,Y,C)$ \{ \\
\scriptsize{1}\> $\pl := \emptyset$; $\Sub{F}{ini}:= F$  \\
\scriptsize{2}\> while (\ti{true}) \{ \\
$........$ \\
\scriptsize{11}$^*$\Tt $D \!:=\!PrvClsRed(\pnt{y},\!F,\!C)$ \\
\scriptsize{12}\Tt $\pl := \pl \cup \s{D}$\}\} \\
\end{tabbing}}
\vspace{-10pt}
\caption{Pseudocode of \egp}
\vspace{10pt}
\label{fig:eg_pqe+}
\end{wrapfigure}

The idea of \egp is to generate a proof for a much weaker proposition
namely a proof of redundancy of $C$ (and only $C$). Intuitively, such
a proof should be easier to generalize. So, \egp calls a procedure
\ti{PrvClsRed} generating such a proof. \egp is a generic algorithm in
the sense that \ti{any} suitable procedure can be employed as
\ti{PrvClsRed}. In our current implementation, the procedure
\dpqe~\cite{hvc-14} is used as \ti{PrvClsRed}. \dpqe generates a proof
stating that $C$ is redundant in \prob{X}{F} in subspace $\pnt{y}^*
\subseteq \pnt{y}$.  Then the plugging clause $D$ falsified by
$\pnt{y}^*$ is generated. Importantly, $\pnt{y}^*$ can be much shorter
than \pnt{y}. Appendix~\ref{app:ds_pqe} gives a brief description of
\dpqe. 

%
%
\begin{example}
\label{exmp:eg_pqe+}
Consider the example solved in Subsection~\ref{ssec:exmp}.  That is,
we consider taking clause $C_1$ out of \prob{X}{F(X,Y)} where $F = C_1
\wedge \dots \wedge C_4$, $C_1=\overline{x}_3 \vee x_4$,
$C_2\!=\!y_1\!\vee\!x_3$, $C_3=y_1 \vee \overline{x}_4$,
$C_4\!=\!y_2\!\vee\!x_4$ and $Y=\s{y_1,y_2}$ and $X=\s{x_3,x_4}$.
Consider the step where \Eg proves redundancy of $C_1$ in subspace
$\pnt{y}=(y_1\!=\!1,y_2\!=\!1)$.  \Eg shows that
$(y_1\!=\!1,y_2\!=\!1,\!x_3=0)$ satisfies $F$, thus proving every
clause of $F$ (including $C_1$) redundant in \prob{X}{F} in subspace
\pnt{y}. Then \Eg generates the plugging clause $D = \overline{y}_1
\vee \overline{y}_2$ falsified by \pnt{y}.

In contrast to \Eg, \egp calls \ti{PrvClsRed} to produce a proof of
redundancy for the clause $C_1$ alone.  Note that $F$ has no clauses
resolvable with $C_1$ on $x_3$ in subspace $\pnt{y}^* = (y_1 =
1)$. (The clause $C_2$ containing $x_3$ is satisfied by $\pnt{y}^*$.)
This means that $C_1$ is blocked in subspace $\pnt{y}^*$ and hence
redundant there (see Proposition~\ref{prop:blk_cls}). Since $\pnt{y}^*
\subset \pnt{y}$, \egp produces a more general proof of redundancy
than \Eg. To avoid re-examining the subspace $\pnt{y}^*$, \egp
generates a \ti{shorter} plugging clause $D = \overline{y}_1$.
\end{example}

%
%
\subsection{Discussion}
\label{ssec:disc2}
Consider the PQE problem of taking a clause $C$ out of
\prob{X}{F(X,Y)}.  There are two features of PQE that make it easier
than QE.  The first feature mentioned earlier is that one can ignore
the subspaces \pnt{y} where $F \setminus \s{C}$ implies $C$. The
second feature is that when \cof{F}{y} is satisfiable, one only needs
to prove redundancy of the clause $C$ alone.  Among the three
algorithms we run in experiments, namely, \dpqe, \Eg, and \egp only
the latter exploits both features. (In addition to using \dpqe inside
\egp we also run it as a stand-alone PQE solver.)  \dpqe does not use
the first feature~\cite{hvc-14} and \Eg does not exploit the second
one. As we show in Sections~\ref{sec:fifo_exper}
and~\ref{sec:inv_gen_exper}, this affects the performance of \dpqe and
\Eg.

\section{Experiment With FIFO Buffers}
\label{sec:fifo_exper}

In this and the next two sections we describe some experiments
with\linebreak \dpqe, \Eg and \egp (their sources are available
at~\cite{ds_pqe},~\cite{eg_pqe} and ~\cite{eg_pqe_plus}
respectively). We used Minisat2.0~\cite{minisat} as an internal
SAT-solver.  The experiments were run on a computer with Intel Core
i5-8265U CPU of 1.6\,GHz.

%
%
\begin{wrapfigure}{l}{1.9in}
\centering
\small
\vspace{-10pt}
\parbox{0cm}{\begin{tabbing}
aa\=bb\=cc\= dd\= \kill
~~~~~~~~~~~$\cdot\cdot\cdot$  \\
if ($\mi{write}==1~\&\&~\mi{currSize}< n$) \\
*\>if ($\mi{dataIn}~!\!= \mi{Val}$) \\
\Tt begin \\
\Tt     $\data[\mi{wrPnt}]  = \mi{dataIn}$; \\
\Tt     $\mi{wrPnt}  = \mi{wrPnt}+1$; \\
\Tt   end \\
~~~~~~~~~~~$\cdot\cdot\cdot$  \\
\end{tabbing}}
\vspace{-25pt}
\caption{A buggy fragment of Verilog code describing \fifo}
\vspace{5pt}
\label{fig:bug}
\end{wrapfigure}
 In this section, we give an example of bug
detection by invariant generation for a FIFO buffer. Our objective
here is threefold.  First, we want to give an example of a bug that
can be overlooked by testing and guessing the unwanted properties to
check (see Subsection~\ref{ssec:hard_bug}). Second, we want to
substantiate the intuition of Subsection~\ref{ssec:des_cov} that
property generation by PQE (in our case, invariant generation by PQE)
has the same reasons to be effective as testing. In particular, by
taking out different clauses one generates invariants relating to
different parts of the design. So, taking out a clause of the buggy
part is likely to produce an unwanted invariant.  Third, we want to
give an example of an invariant that can be easily identified as
unwanted\footnote{Let $P(\hat{S})$ be an invariant for a circuit $N$ depending only on a
subset $\hat{S}$ of the state variables $S$. Identifying $P$ as an
unwanted invariant is much easier if $\hat{S}$ is meaningful from the
high-level view of the design.  Suppose, for instance, that
assignments to $\hat{S}$ specify values of a high-level variable
$v$. Then $P$ is unwanted if it claims unreachability of a value of
$v$ that is supposed to be reachable. Another simple example is that
assignments to $\hat{S}$ specify values of high-level variables $v$
and $w$ that are supposed to be \ti{independent}. Then $P$ is unwanted
if it claims that some combinations of values of $v$ and $w$ are
unreachable. (This may mean, for instance, that an assignment operator
setting the value of $v$ erroneously involves the variable $w$.)
}.
%
%
\subsection{Buffer description}
\label{ssec:buff_descr}
Consider a FIFO buffer that we will
refer to as \fifo.  Let $n$ be the number of elements of \fifo and
\data denote the data buffer of \fifo.  Let each
$\data[i],i=1,\dots,n$ have $p$ bits and be an integer where $0 \leq
\data[i] < 2^p$.  A fragment of the Verilog code describing \fifo is
shown in Fig~\ref{fig:bug}. This fragment has a buggy line marked with
an asterisk. In the correct version without the marked line, a new
element $\mi{dataIn}$ is added to \data if the \ti{write} flag is on
and \fifo has less than $n$ elements.  Since \data can have any
combination of numbers, all \data states are supposed to be reachable.
However, due to the bug, the number $\mi{Val}$ cannot appear in \data.
(Here $\mi{Val}$ is some constant \mbox{$0\!<\!\mi{Val}\!<\!
  2^p$}. We assume that the buffer elements are initialized to 0.) So,
\fifo has an \ti{op-state reachability bug} since it cannot reach
operative states where an element of \data equals $\mi{Val}$.

%
%
\subsection{Bug detection by invariant generation}
Let $N$ be a circuit implementing \fifo. Let $S$ be the set of state
variables of $N$ and \bm{\Sub{S}{data} \subset S} be the subset
corresponding to the data buffer \data.  We used \dpqe, \Eg and \egp
to generate invariants of $N$ as described in
Section~\ref{sec:inv_gen}.  Note that an invariant $Q$ depending only
on \Sub{S}{data} is an \tb{unwanted} one.  If $Q$ holds for $N$, some
states of \data are unreachable. Then \fifo has an op-state
reachability bug since every state of \data is supposed to be
reachable. To generate invariants, we used the formula $F_k = I(S_0)
\wedge T(S_0,V_0,S_1) \wedge \dots \wedge T(S_{k-1},V_{k-1},S_k)$
introduced in Subsection~\ref{ssec:check_unreach}. Here $I$ and $T$
describe the initial state and the transition relation of $N$
respectively and $S_j$ and $V_j$ denote state variables and
combinational input variables of $j$-th time frame respectively.
First, we used a PQE solver to generate a local invariant $H(S_k)$
obtained by taking a clause $C$ out of \prob{\Abs{k}}{F_k} where
$\Abs{k} = S_0 \cup V_0 \cup \dots \cup S_{k-1} \cup V_{k-1}$.  So,
\mbox{$\prob{\Abs{k}}{F_k}\equiv$} $H \wedge$ \prob{\Abs{k}}{F_k
  \setminus \s{C}}. (Since $F_k \imp H$, no state falsifying $H$ can
be reached in $k$ transitions.)  In the experiment, we took out only
clauses of $F_k$ containing an \ti{unquantified variable}, \ie a state
variable of the $k$-th time frame. The time limit for solving the PQE
problem of taking out a clause was set to 10 sec.

\begin{table}[h]
\centering
\vspace{4pt}
\scriptsize
\caption{\small{FIFO buffer with $n$ elements of 32 bits. Time
    limit is 10 sec. per PQE problem}}
\begin{tabular}{|p{20pt}|p{18pt}|p{21pt}|p{20pt}|p{20pt}|p{18pt}|p{22pt}|p{22pt}|p{22pt}|p{17pt}|p{17pt}|p{20pt}|p{22pt}|p{22pt}|p{20pt}|} \hline
 buff.     & lat-  &  time  &\multicolumn{3}{c|}{total \ti{pqe} probs}  &\multicolumn{3}{c|}{finished \ti{pqe} probs}  &\multicolumn{3}{c|}{unwant. invar}  &\multicolumn{3}{c|}{runtime (s.)} \\ \cline{4-15}
size & ches & fra- & \ti{ds-} &\ti{eg-} &\ti{eg-}& \ti{ds-} &\ti{eg-} & \ti{eg-} & \ti{ds-} & \ti{eg-} & \ti{eg-} & \ti{ds-} & \ti{eg-}&\ti{eg-} \\ 
~~$n$ & &mes & \ti{pqe} & \ti{pqe} & \ti{pqe}$^+$    &   \ti{pqe}  &   \ti{pqe}   &   \ti{pqe}$^+$     &    \ti{pqe}      &  \ti{pqe}         &   \ti{pqe}$^+$            &    \ti{pqe}      &       \ti{pqe}   &  \ti{pqe}$^+$            \\ \hline
 ~~8    & ~300      &~ 5  &  1,236  &  311    & ~\tb{8}   &~2\%  & \tb{36\%} & 35\% &no &\tb{yes}&\tb{yes}&12,141&2,138  &\tb{52} \\ \hline
  ~~8    & ~300      &~ 10  & ~560   & 737     &~\tb{39}  &~2\%  & 1\% & \tb{3\%}  &\tb{yes} &\tb{yes}&\tb{yes}&5,551&7,681    &\tb{380}         \\ \hline
  ~~16    & ~560      &~ 5  & 2,288   & 2,288  &~\tb{16}   &~1\% & 65\% & \tb{71\%}  &no &no&\tb{yes}& 22,612    & 9,506    &\tb{50}        \\ \hline
  ~~16   & ~560      &~ 10  & 653 &  2,288   &~\tb{24} &~1\%  & 36\% & \tb{38\%}  &\tb{yes} &no&\tb{yes}&6,541  &16,554&\tb{153}      \\ \hline  
\end{tabular}                
\vspace{5pt}
\label{tbl:buff}
\end{table}

For each clause $Q$ of every local invariant $H$ generated by PQE, we
checked if $Q$ was a global invariant. Namely, we used a public
version of \ict~\cite{ic3,ic3_impl} to verify if the property $Q$ held
(by showing that no reachable state of $N$ falsified $Q$).  If so, and
$Q$ depended only on variables of \Sub{S}{data}, $N$ had an
\ti{unwanted invariant}. Then we stopped invariant generation. The
results of the experiment for buffers with 32-bit elements are given
in Table~\ref{tbl:buff}.  When picking a clause to take out, \ie a
clause with a state variable of $k$-th time frame, one could make a
good choice by pure luck.  To address this issue, we picked clauses to
take out \ti{randomly} and performed 10 different runs of invariant
generation and then computed the average value. So, the columns four
to twelve of Table~\ref{tbl:buff} actually give the average value of
10 runs.

Let us use the first line of Table~\ref{tbl:buff} to explain its
structure. The first two columns show the number of elements in \fifo
implemented by $N$ and the number of latches in $N$ (8 and 300).  The
third column gives the number $k$ of time frames (\ie 5). The next
three columns show the total number of PQE problems solved by a PQE
solver before an unwanted invariant was generated, \eg 8 problems for
\egp. On the other hand, \dpqe failed to find an unwanted invariant
and had to solve \ti{all} 1,236 PQE problems of taking out a clause of
$F_k$ with an unquantified variable. The following three columns show
the share of PQE problems \ti{finished} in the time limit of 10 sec.
For instance, \Eg finished 36\% of 311 problems.  The next three
columns show if an unwanted invariant was generated by a PQE solver.
(\Eg and \egp found one whereas \dpqe did not.) The last three columns
give the total run time.  Table~\ref{tbl:buff} shows that only \egp
managed to generate an unwanted invariant for all four instances of
\fifo. This invariant asserted that \fifo cannot reach a state where
an element of \data equals $\mi{Val}$.
%
%
\subsection{Detection of the bug by conventional methods}
\label{ssec:hard_bug}
The bug above (or its modified version) can be overlooked by
conventional methods.  Consider, for instance, testing.  It is hard to
detect this bug by \ti{random} tests because it is exposed only if one
tries to add \ti{Val} to \fifo.  The same applies to testing using the
\ti{line coverage} metric~\cite{coverage}. On the other hand, a test
set with 100\% \ti{branch} coverage~\cite{coverage} will find this
bug. (To invoke the \ti{else} branch of the \ti{if} statement marked
with '*' in Fig.~\ref{fig:bug}, one must set $\mi{dataIn}$ to
$\mi{Val}$.)  However, such a test set is not hard to beat. Consider a
slightly modified bug. Namely, assume that the buggy line
is~~\ti{if}\,(($\mi{dataIn}~!\!\!= \mi{Val})~ \&\&
~(\mi{dataIn}~!\!\!= \mi{Val}')$) and in the \ti{else} branch of this
\ti{if} statement, the element $\mi{Val}'$ is pushed to \fifo. So,
instead of just ignoring the element $\mi{Val}$, the modified bug
\ti{replaces} it with $\mi{Val}'$ in \fifo.  Similarly to the bug of
Subsection~\ref{ssec:buff_descr}, $\mi{Val}$ does not appear in
\fifo. So, the modified bug can also be detected by generating an
unwanted invariant. On the other hand, it can be overlooked by tests
with 100\% branch coverage because either branch of the modified
\ti{if} statement can be activated \ti{without} assigning
$\mi{dataIn}$ to $\mi{Val}$.

Now consider the ``manual'' generation of unwanted properties. It is
virtually impossible to guess an unwanted \ti{invariant} of \fifo
exposing the bug of Subsection~\ref{ssec:buff_descr} unless one knows
exactly what this bug is.  However, one can detect this bug by
checking a property asserting that the element $\mi{dataIn}$ must
appear in the buffer if \fifo is ready to accept it, \ie  $(\mi{write}
== 1) \&\& (\mi{currSize} < n)$ holds. Note that this is a
\ti{non-invariant} property involving states of different time
frames. The more time frames are used in such a property, the more
guesswork is required to pick it.
Consider, for instance, the following bug. Suppose \fifo does not
reject the element $\mi{Val}$.  So, the non-invariant property above
holds.  However, if $\mi{dataIn} == \mi{Val}$, then \fifo changes the
\ti{previous} accepted element if that element was $\mi{Val}$ too. So,
\fifo cannot have two consecutive elements $\mi{Val}$.  Our method
will detect this bug via generating an unwanted invariant falsified by
states with consecutive elements $\mi{Val}$. One can also identify
this bug by checking a property involving two consecutive elements of
\fifo. But picking it requires a lot of guesswork and so the modified
bug can be easily overlooked.

\section{Experiments With HWMCC Benchmarks}
\label{sec:inv_gen_exper}

In this section, we describe three experiments with 98 multi-property
benchmarks of the HWMCC-13 set~\cite{hwmcc13}.  (We use this set
because it has a multi-property track, see the explanation below.)
The number of latches in those benchmarks range from 111 to 8,000.
More details about the choice of benchmarks and the experiments can be
found in Appendix~\ref{app:exper2}. Each benchmark consists of a
sequential circuit $N$ and invariants $P_0,\dots,P_m$ to prove. Like
in Section~\ref{sec:inv_gen}, we call $\Sub{P}{agg}=P_0 \wedge \dots
\wedge P_m$ the \ti{aggregate invariant}.  In experiments 2 and 3 we
used PQE to generate new invariants of $N$. Since every invariant $P$
implied by \Sub{P}{agg} is a desired one, the necessary condition for
$P$ to be \ti{unwanted} is $\Sub{P}{agg} \not\imp P$. The conjunction
of many invariants $P_i$ produces a stronger invariant \Sub{P}{agg},
which makes it \ti{harder} to generate $P$ not implied by
\Sub{P}{agg}. (This is the reason for using multi-property benchmarks
in our experiments.)  The circuits of the HWMCC-13 set are
\ti{anonymous}, so, we could not know if an unreachable state is
supposed to be reachable.  For that reason, we just generated
invariants not implied by \Sub{P}{agg} without deciding if some of
them were unwanted.

Similarly to the experiment of Section~\ref{sec:fifo_exper}, we used
the formula $F_k = I(S_0) \wedge T(S_0,V_0,S_1)$ $\wedge \dots \wedge
T(S_{k-1},V_{k-1},S_k)$ to generate invariants. The number $k$ of time
frames was in the range of \mbox{$2\!\leq\!k\!  \leq\!10$}.  As in the
experiment of Section~\ref{sec:fifo_exper}, we took out only clauses
containing a state variable of the $k$-th time frame. In all
experiments, the \tb{time limit} for solving a PQE problem was set to
10 sec.

%
%
\subsection{Experiment 1}
In the first experiment, we generated a \ti{local invariant} $H$ by
taking out a clause $C$ of \prob{\Abs{k}}{F_k} where $\Abs{k}=S_0 \cup
V_0 \cup \dots \cup S_{k-\!1}\cup V_{k-\!1}$.  The formula $H$ asserts
that no state falsifying $H$ can be reached in $k$ transitions. Our
goal was to show that PQE can find $H$ for large formulas $F_k$ that
have hundreds of thousands of clauses.  We used \Eg to partition the
PQE problems we tried into two groups.  \ti{The first group} consisted
of 3,736 problems for which we ran \Eg with the time limit of 10
sec. and it never encountered a subspace \pnt{s_k} where $F_k$ was
satisfiable.  Here \pnt{s_k} is a full assignment to $S_k$. Recall
that only the variables $S_k$ are unquantified in
\prob{\Abs{k}}{F_k}. So, in every subspace \pnt{s_k}, formula $F_k$
was either unsatisfiable or $(F_k \setminus \s{C}) \imp C$. (The fact
that so many problems meet the condition of the first group came as a
big surprise.) \ti{The second group} consisted of 3,094 problems where
\Eg encountered subspaces where $F_k$ was satisfiable.

For the first group, \dpqe finished only 30\% of the problems within
10 sec. whereas \Eg and \egp finished 88\% and 89\% respectively. The
poor performance of \dpqe is due to not checking if $(F_k \setminus
\s{C}) \imp C$ in the current subspace. For the second group, \dpqe,
\Eg and \egp finished 15\%, 2\% and 27\% of the problems respectively
within 10 sec. \Eg finished far fewer problems because it used a
satisfying assignment as a proof of redundancy of $C$ (see
Subsection~\ref{ssec:disc2}).

To contrast PQE and QE, we employed a high-quality tool
\cad~\cite{cadet_qe,cadet_imp} to perform QE on the 98 formulas
\prob{\Abs{k}}{F_k} (one formula per benchmark).  That is, instead of
taking a clause out of \prob{\Abs{k}}{F_k} by PQE, we applied \cad to
perform full QE on this formula. (Performing QE on \prob{\Abs{k}}{F_k}
produces a formula $H(S_k)$ specifying \ti{all} states unreachable in
$k$ transitions.)  \cad finished only 25\% of the 98 QE problems with
the time limit of 600 sec. On the other hand, \egp finished 60\% of
the 6,830 problems of both groups (generated off \prob{\Abs{k}}{F_k})
within 10 sec. So, PQE can be much easier than QE if only a small part
of the formula gets unquantified.

%
%
\subsection{Experiment 2}
The second experiment was an extension of the first one. Its goal was
to show that PQE can generate invariants for realistic designs.  For
each clause $Q$ of a local invariant $H$ generated by PQE we used \ict
to verify if $Q$ was a global invariant.  If so, we checked if
$\Sub{P}{agg} \not\imp Q$ held. To make the experiment less time
consuming, in addition to the time limit of 10 sec. per PQE problem we
imposed a few more constraints.  The PQE problem of taking a clause
out of \prob{\Abs{k}}{F_k} terminated as soon as $H$ accumulated 5
clauses or more. Besides, processing a benchmark aborted when the
summary number of clauses of all formulas $H$ generated for this
benchmark reached 100 or the total run time of all PQE problems
generated off \prob{\Abs{k}}{F_k} exceeded 2,000 sec.

%
%
\begin{wraptable}{l}{2.1in}
\centering
\vspace{2pt}
\scriptsize
\captionsetup{justification=centering}
\caption{\small{Invariant generation}}
\vspace{-4pt}
  \begin{tabular}{|p{28pt}|p{26pt}|p{20pt}|p{20pt}|p{30pt}|} \hline
 pqe &\#bench  & \multicolumn{3}{c|}{results} \\ \cline{3-5}
 solver &marks   & local &glob. & not imp.\\
        &  & invar.  &invar. & by $P_{agg}$\\ \hline
\ti{ds-pqe}     &~98   & 5,556   & 2,678 &2,309\\ \hline
\ti{eg-pqe}     &~98   & \tb{9,498}   & \tb{4,839} & \tb{4,009}\\ \hline
\ti{eg-pqe}$^+$ &~98   & 9,303   & 4,773 & 3,940\\ \hline 
\end{tabular}                
\vspace{4pt}
\label{tbl:all_inv_gen}
\end{wraptable}

Table~\ref{tbl:all_inv_gen} shows the results of the experiment. The
third column gives the number of local single-clause invariants (\ie
the total number of clauses in all $H$ over all benchmarks). The
fourth column shows how many local single-clause invariants turned out
to be global. (Since global invariants were extracted from $H$ and the
summary size of all $H$ could not exceed 100, the number of global
invariants per benchmark could not exceed 100.) The last column gives
the number of global invariants not implied by \spe. So, these
invariants are candidates for checking if they are
unwanted. Table~\ref{tbl:all_inv_gen} shows that \Eg and \egp
performed much better than \dpqe.

%
%
\subsection{Experiment 3}
\label{ssec:ic3_invars}
To prove an invariant $P$ true, \ict conjoins it with clauses
$Q_1\!,\dots,\!Q_n$ to make $P\!\wedge Q_1\!\wedge \dots \wedge Q_n$
inductive.  If \ict succeeds, every $Q_i$ is an invariant. Moreover,
$Q_i$ may be an \ti{unwanted} invariant. The goal of the third
experiment was to demonstrate that PQE and \ict, in general, produce
different invariant clauses. The intuition here is twofold. First,
\ict generates clauses $Q_i$ to prove a \ti{predefined} invariant
rather than find an unwanted one. Second, the closer $P$ to being
inductive, the fewer new invariant clauses are generated by
\ict. Consider the circuit \Sub{N}{triv} that simply stays in the
initial state \pqnt{s}{ini} (Section~\ref{sec:inv_gen}). Any invariant
satisfied by \pqnt{s}{ini} is already \ti{inductive} for
\Sub{N}{triv}. So, IC3 will not generate \ti{a single new invariant}
clause. On the other hand, if the correct circuit is supposed to leave
the initial state, \Sub{N}{triv} has unwanted invariants that our
method will find.

In this experiment, we used \ict to generate $P^*_{\mi{agg}}$, an
\ti{inductive} version of \Sub{P}{agg}. The experiment showed that in
88\% cases, an invariant clause generated by \egp and not implied by
\Sub{P}{agg} was not implied by $P^*_{\mi{agg}}$ either.  (See
Appendix~\ref{ssec:exper3} for more detail.) 

\section{Properties Mimicking Symbolic Simulation}
\label{sec:comb_exper}
Let $M(X,V,W)$ be a combinational circuit where $X,V,W$ are internal,
input and output variables.  In this section, we describe generation
of properties of $M$ that mimic symbolic
simulation~\cite{SymbolSim}. Every such a property $Q(V)$ specifies a
cube of tests that produce the same values for a given subset of
variables of $W$. We chose generation of such properties because
deciding if $Q$ is an unwanted property is, in general, simple. The
procedure for generation of these properties is slightly different
from the one presented in Section~\ref{sec:prop_gen}.

Let $F(X,V,W)$ be a formula specifying $M$. Let $B(W)$ be a clause.
Let $H(V)$ be a solution to the PQE problem of taking a clause $C \in
F$ out of \Prob{X}{W}{F \wedge B}. That is, $\Prob{X}{W}{F \wedge B}
\equiv H \wedge$ \Prob{X}{W}{(F \setminus \s{C}) \wedge B}. Let $Q(V)$
be a clause of $H$. Then $M$ has the \tb{property} that for every full
assignment \pnt{v} to $V$ falsifying $Q$, it produces an output
\pnt{w} falsifying $B$ (see Proposition~\ref{prop:symb_sim} of
Appendix~\ref{app:proofs}). Suppose, for instance, $Q\!=\!v_1 \vee\,
\overline{v}_{10} \vee v_{30}$ and $B\!=\!w_2 \vee\,
\overline{w}_{40}$. Then for every \pnt{v} where
$v_1\!=\!0,v_{10}\!=\!1,\!v_{30}\!=\!0$, the circuit $M$ produces an
output where $w_2 = 0, w_{40}=1$.  Note that $Q$ is implied by $F
\wedge B$ rather than $F$. So, it is a property of $M$ under
constraint $B$ rather than $M$ alone. The property $Q$ is
\tb{unwanted} if there is an input falsifying $Q$ that \ti{should not}
produce an output falsifying $B$.

To generate combinational circuits, we unfolded sequential circuits of
the set of 98 benchmarks used in Section~\ref{sec:inv_gen_exper} for
invariant generation.  Let $N$ be a sequential circuit. (We reuse the
notation of Section~\ref{sec:inv_gen}).  Let
$M_k(S_0,V_0,\dots,S_{k-1},V_{k-1},S_k)$ denote the combinational
circuit obtained by unfolding $N$ for $k$ time frames. Here $S_j,V_j$
are state and input variables of $j$-th time frame respectively. Let
$F_k$ denote the formula $I(S_0) \wedge T(S_0,V_0,S_1) \wedge \dots
\wedge T(S_{k-1},V_{k-1},S_k)$ describing the unfolding of $N$ for $k$
time frames. Note that $F_k$ specifies the circuit $M_k$ above under
the input constraint $I(S_0)$. Let $B(S_k)$ be a clause. Let
$H(S_0,V_0,\dots,V_{k-1})$ be a solution to the PQE problem of taking
a clause $C \in F_k$ out of formula \prob{S_{1,k}}{F_k \wedge B}. Here
$S_{1,k} = S_1 \cup \dots \cup S_k$.  That is, $\prob{S_{1,k}}{F_k
  \wedge B} \equiv H \wedge$ \prob{S_{1,k}}{(F_k \setminus \s{C})
  \wedge B}. Let $Q$ be a clause of $H$. Then for every assignment
(\pqnt{s}{ini},\pnt{v_0},\dots,\ppnt{v}{k-1}) falsifying $Q$, the
circuit $M_k$ outputs \pent{s}{k} falsifying $B$. (Here \pqnt{s}{ini}
is the initial state of $N$ and \pent{s}{k} is a state of the last
time frame.)

In the experiment, we used \dpqe,\Eg and \egp to solve 1,586 PQE
problems described above.  In Table~\ref{tbl:symb_sim}, we give a
sample of results by \egp. (More details about this experiment can be
found in Appendix~\ref{app:symb_sim}.) Below, we use the first line of
Table~\ref{tbl:symb_sim} to explain its structure. The first column
gives the benchmark name (6s326). The next column shows that 6s326 has
3,342 latches. The third column gives the number of time frames used
to produce a combinational circuit $M_k$ (here $k=20$). The next
column shows that the clause $B$ introduced above consisted of 15
literals of variables from $S_k$.  (Here and below we still use the
index $k$ assuming that $k = 20$.)  The literals of $B$ were generated
\ti{randomly}. When picking the length of $B$ we just tried to
simulate the situation where one wants to set a particular \ti{subset}
of output variables of $M_k$ to specified values.  The next two
columns give the size of the subcircuit $M'_k$ of $M_k$ that feeds the
output variables present in $B$.  When computing a property $H$ we
took a clause out of formula \prob{S_{1,k}}{F'_k \wedge B} where
$F'_k$ specifies $M'_k$ instead of formula \prob{S_{1,k}}{F_k \wedge
  B} where $F_k$ specifies $M_k$. (The logic of $M_k$ not feeding a
variable of $B$ is irrelevant for computing $H$.)  The first column of
the pair gives the number of gates in $M'_k$ (\ie 348,479). The second
column provides the number of input variables feeding $M'_k$ (\ie
1,774).  Here we count only variables of $V_0 \cup \dots \cup V_{k-1}$
and ignore those of $S_0$ since the latter are already assigned values
specifying the initial state \pqnt{s}{ini} of $N$.

%
%
\begin{wraptable}{l}{3.1in}
\centering
\scriptsize
\captionsetup{justification=centering}
\caption{\small{Property generation for combinational \\circuits}}
  \begin{tabular}{|p{22pt}|p{18pt}|p{16pt}|p{15pt}|p{27pt}|p{22pt}|p{15pt}|p{15pt}|p{16pt}|p{20pt}|} \hline
name & lat- &time &size & \multicolumn{2}{c|}{subc. $M'_k$} &  \multicolumn{4}{c|}{results}\\  \cline{5-10}
     &ches &  fra- & of   & gates     & inp.    &  min &max & time&3-val.\\ 
      &  & mes   & $B$         &          & vars        &     &    & (s.)&sim.\\ \hline
 6s326&3,342  &~20   &~15     & 348,479  &1,774   &27&28& 2.9 &~\tb{no}  \\ \hline
 6s40m& 5,608 &~20  &~15     &406,474 &3,450  &27 &29 & 1.1 &~\tb{no}  \\ \hline
 6s250& 6,185  &~20    &~15    & 556,562 &2,456  &50 &54 & 0.8 &~\tb{no} \\ \hline
 6s395 & 463  &~30   &~15     &36,088 & 569   &24&26 &0.7  &~yes \\ \hline
 6s339& 1,594 &~30   &~15     &179,543 &3,978   &70 &71 & 3.1 &~\tb{no}\\ \hline
 6s292 & 3,190  &~30    &~15   & 154,014  & 978  &86 &89 & 1.1 &~\tb{no}  \\ \hline
 6s143 & 260  &~40   &~15  & 551,019   & 16,689  &526 &530 & 2.5 &~yes  \\ \hline
 6s372 &1,124  &~40   &~15       &295,626&  2,766 &513 &518 &1.7 &~\tb{no} \\ \hline
  6s335 &1,658  &~40   &~15   & 207,787& 2,863  &120 &124  & 6.7  &~\tb{no} \\ \hline 
 6s391  &2,686 &~40   &~15     & 240,825& 7,579 &340 & 341& 8.9 &~\tb{no}  \\ \hline
\end{tabular}                
\label{tbl:symb_sim}
\end{wraptable}

The next four columns show the results of taking a clause out
of\linebreak \prob{S_{1,k}}{F'_k\!\wedge\!B}. For each PQE problem the
time limit was set to 10 sec. Besides, \egp terminated as soon as 5
clauses of property\linebreak $H(S_0,V_0,\dots,V_{k-1})$ were
generated.  The first three columns out of four describe the minimum
and maximum sizes of clauses in $H$ and the run time of \egp. So, it
took for \egp 2.9 sec.  to produce a formula $H$ containing clauses of
sizes from 27 to 28 variables. A clause $Q$ of $H$ with 27 variables,
for instance, specifies $2^{1747}$ tests falsifying $Q$ that produce
the same output of $M'_k$ (falsifying the clause $B$). Here
\mbox{$1747 = 1774-27$} is the number of input variables of $M'_k$ not
present in $Q$. The last column shows that at least one clause $Q$ of
$H$ specifies a property that cannot be produced by 3-valued
simulation (a version of symbolic simulation~\cite{SymbolSim}).  To
prove this, one just needs to set the input variables of $M'_k$
present in $Q$ to the values falsifying $Q$ and run 3-valued
simulation. (The remaining input variables of $M'_k$ are assigned a
don't-care value.) If after 3-valued simulation some output variable
of $M'_k$ is assigned a don't-care value, the property specified by
$Q$ cannot be produced by 3-valued simulation.

Running \dpqe, \Eg and \egp on the 1,586 PQE problems mentioned above
showed that a) \Eg performed poorly producing properties only for 28\%
of problems; b) \dpqe and \egp showed much better results by
generating properties for 62\% and 66\% of problems respectively.
When \dpqe and \egp succeeded in producing properties, the latter
could not be obtained by 3-valued simulation in 74\% and 78\% of cases
respectively.

\section{Some Background}
\label{sec:bg}
In this section, we discuss some research relevant to PQE and property
generation.  Information on BDD based QE can be found
in~\cite{bryant_bdds1,bdds_qe}. SAT based QE is described in
\cite{blocking_clause,fabio,cofactoring,cav09,cav11,cmu,nik1,nik2,cadet_qe}.
Our first PQE solver called \dpqe was introduced in~\cite{hvc-14}. It
was based on redundancy based reasoning presented in~\cite{fmcad12} in
terms of variables and in~\cite{fmcad13} in terms of clauses. The main
flaw of \dpqe is as follows. Consider taking a clause $C$ out of
\prob{X}{F}. Suppose \dpqe proved $C$ redundant in a subspace where
$F$ is \ti{satisfiable} and some \ti{quantified} variables are
assigned. The problem is that \dpqe cannot simply assume that $C$ is
redundant every time it re-enters this subspace~\cite{qe_learn}. The
root of the problem is that redundancy is a \ti{structural} rather
than semantic property. That is, redundancy of a clause in a formula
$\xi$ (quantified or not) does not imply such redundancy in every
formula logically equivalent to $\xi$. Since our current
implementation of \egp uses \dpqe as a subroutine, it has the same
learning problem.  We showed in~\cite{cert_tech_rep} that this problem
can be addressed by the machinery of certificate clauses. So, the
performance of PQE can be drastically improved via enhanced learning
in subspaces where $F$ is satisfiable.

We are unaware of research on property generation for combinational
circuits. As for invariants, the existing procedures typically
generate some auxiliary \ti{desired} invariants to prove a predefined
property (whereas our goal is to generate invariants that are
\ti{unwanted}). For instance, they generate loop
invariants~\cite{loop_invars} or invariants relating internal points
of circuits checked for equivalence ~\cite{ec_invars}. Another example
of auxiliary invariants are clauses generated by \ict to make an
invariant inductive~\cite{ic3}.  As we showed in
Subsection~\ref{ssec:ic3_invars}, the invariants produced by PQE are,
in general, different from those built by \ict.

\section{Conclusions And Directions For Future Research}
\label{sec:concl}
We consider Partial Quantifier Elimination (PQE) on propositional CNF
formulas with existential quantifiers. In contrast to \ti{complete}
quantifier elimination, PQE allows to unquantify a \ti{part} of the
formula.  We show that PQE can be used to generate properties of
combinational and sequential circuits. The goal of property generation
is to check if there is an \ti{unwanted} property identifying a bug.
We used PQE to generate an unwanted invariant for a FIFO buffer
exposing a non-trivial bug.  We also applied PQE to invariant
generation for HWMCC benchmarks. Finally, we used PQE to generate
properties of combinational circuits mimicking symbolic simulation.
Our experiments show that PQE can efficiently generate properties for
realistic designs.

There are at least three directions for future research. The first
direction is to improve the performance of PQE solving. As we
mentioned in Section~\ref{sec:bg}, the most promising idea here is to
enhance the power of learning in subspaces where the formula is
satisfiable.  The second direction is to use the improved PQE solvers
to design new, more efficient algorithms for well-known problems like
SAT, model checking and equivalence checking. The third direction is
to look for new problems that can be solved by PQE.

\bibliographystyle{IEEEtran}
\bibliography{short_sat,local,l1ocal_hvc}
\vspace{25pt}
\appendix
\noindent{\large \tb{Appendix}}
\section{PQE And Interpolation}
\label{app:interp}
In this appendix, we recall the observation of~\cite{tech_rep_pc_lor}
that interpolation is a special case of PQE. Let $A(X,Y) \wedge
B(Y,Z)$ be an unsatisfiable formula. Let $I(Y)$ be a formula such that
$A \wedge B \equiv I \wedge B$ and $A \imp I$. Then $I$ is called an
\ti{interpolant}~\cite{craig}.  Now, let us show that interpolation
can be described in terms of PQE. Consider the formula \prob{W}{A
  \wedge B} where $A$ and $B$ are the formulas above and $W = X \cup
Z$. Let $A^*(Y)$ be obtained by taking $A$ out of the scope of
quantifiers i.e. \mbox{\prob{W}{A \wedge B} $\equiv
  A^*\wedge$\prob{W}{B}}. Since $A \wedge B$ is unsatisfiable, $A^*
\wedge B$ is unsatisfiable too. So, \mbox{$A\wedge B \equiv A^* \wedge
  B$}. If $A \imp A^*$, then $A^*$ is an interpolant.

The \ti{general case} of PQE that takes $A$ out of \prob{W}{A \wedge
  B} is different from the instance above in three aspects. First, one
does not assume that $A \wedge B$ is unsatisfiable.  Second, one does
not assume that $\V{B} \subset \V{A \wedge B}$. In other words, in
general, PQE \ti{does not} remove any variables from the original
formula.  Third, a solution $A^*$ is implied by $A \wedge B$ rather
than by $A$ alone.  Summarizing, one can say that interpolation is a
special case of PQE.

\section{Proofs Of Propositions}
 \setcounter{proposition}{0}
 \label{app:proofs}
 %
 %
\begin{proposition}
Let $H$ be a solution to the PQE problem of
Definition~\ref{def:pqe_prob}.  That is $\prob{X}{F}\equiv
H\wedge\prob{X}{F \setminus G}$. Then $F \imp H$ (i.e. $F$ implies
$H$).
\end{proposition}
\begin{proof}
By conjoining both sides of the equality with $H$ one concludes
that\linebreak $H \wedge \prob{X}{F} \equiv
H\wedge\prob{X}{F \setminus G}$ and hence
$H \wedge \prob{X}{F} \equiv \prob{X}{F}$. Then $\prob{X}{F} \imp H$
and thus $F \imp H$.
\end{proof}
%
%
\begin{proposition}
Let a clause $C$ be blocked in a formula $F(X,Y)$ with respect to a
variable $x \in X$.  Then $C$ is redundant in \prob{X}{F}
i.e. \prob{X}{F \setminus \s{C}} $\equiv$ \prob{X}{F}.
\end{proposition}
\begin{proof}
It was shown in~\cite{blocked_clause} that adding a clause $B(X)$
blocked in $G(X)$ to the formula \prob{X}{G} does not change the value
of this formula.  This entails that removing a clause $B(X)$ blocked
in $G(X)$ does not change the value of \prob{X}{G} either. So, $B$ is
redundant in \prob{X}{G}. Let \pnt{y} be a full assignment to
$Y$. Then the clause $C$ is either satisfied by \pnt{y} or \cof{C}{y}
is blocked in \cof{F}{y} with respect to $x$. (The latter follows from
the definition of a blocked clause.) In either case \cof{C}{y} is
redundant in \prob{X}{\cof{F}{y}}. Since \cof{C}{y} is redundant in
\prob{X}{\cof{F}{y}} in every subspace \pnt{y}, $C$ is redundant in
\prob{X}{F}.
\end{proof}
%
%
\begin{proposition}
\label{prop:symb_sim}
Let $M(X,V,W)$ be a combinational circuit where $X,V,W$ are the
internal, input and output variables. Let $F(X,V,W)$ be a formula
specifying $M$. Let $B(W)$ be a clause.  Let $H(V)$ be a formula
obtained by taking a clause $C \in F$ out of \Prob{X}{W}{F \wedge
  B}. That is \Prob{X}{W}{F \wedge B} $\equiv H \wedge$ \Prob{X}{W}{(F
  \setminus \s{C}) \wedge B}. Let $Q(V)$ be a clause of $H$. Then for
every full assignment \pnt{v} to $V$ falsifying $Q$, the circuit $M$
outputs an assignment \pnt{w} falsifying the clause $B$.
\end{proposition}
\begin{proof}
  From Proposition~\ref{prop:sol_impl} it follows that $F \wedge B
  \imp H$ and hence $F \wedge B \imp Q$. This entails that
  $\overline{Q} \imp \overline{B} \vee \overline{F}$. Let \pnt{v} be a
  full assignment to $V$ i.e. an input to $M$. Let
  (\pnt{x},\pnt{v},\pnt{w}) be the execution trace produced by $M$
  under the input \pnt{v}. Here \pnt{x},\pnt{w} are full assignments
  to $X$ and $W$ respectively. Suppose, \pnt{v} satisfies
  $\overline{Q}$ (and so falsifies $Q$). Then
  (\pnt{x},\pnt{v},\pnt{w}) satisfies $\overline{B} \vee
  \overline{F}$.  Since (\pnt{x},\pnt{v},\pnt{w}) is an execution
  trace, it satisfies $F$ and so falsifies $\overline{F}$.  This
  entails that (\pnt{x},\pnt{v},\pnt{w}) (and specifically \pnt{w})
  satisfies $\overline{B}$ and hence falsifies $B$.
\end{proof}

\section{Examples Of Problems That Reduce To PQE}
\label{app:using_pqe}

\noindent\tb{SAT-solving by PQE~\nf{\cite{hvc-14}}.}
Consider the SAT problem of checking if
formula \prob{X}{F(X)} is true. One can view traditional SAT-solving
as proving \ti{all} clauses redundant in \prob{X}{F} e.g. by finding a
satisfying assignment or by deriving an empty clause and adding it to
$F$. The reduction to PQE below facilitates developing an incremental
SAT-algorithm that needs to prove redundancy only for a \ti{fraction}
of clauses.

Let \pnt{x} be a full assignment to $X$ and $G$ denote the clauses of
$F$ falsified by \pnt{x}. Checking the satisfiability of $F$ reduces
to taking $G$ out of the scope of quantifiers i.e. to finding $H$ such
that \mbox{$\prob{X}{F} \equiv H \wedge \prob{X}{F \setminus
    G}$}. Since all variables of $F$ are quantified in \prob{X}{F},
the formula $H$ is a Boolean constant 0 or 1. If \mbox{$H=0$}, then
$F$ is unsatisfiable. If $H\!=\!1$, then $F$ is satisfiable because $F
\setminus G$ is satisfied by \pnt{x}.

%
%
\vspace{4pt}
\noindent\tb{Equivalence checking by PQE~\nf{\cite{fmcad16}}.}  Let
$N'(X',V',w')$ and $N''(X'',V'',w'')$ be single-output combinational
circuits to check for equivalence. Here $X',V'$ are the sets of
internal and input variables and $w'$ is the output variable of
$N'$. (Definition of $X'',V'',w''$ for $N''$ is the same.) The
reduction to PQE below facilitates the design of a \ti{complete}
algorithm able to exploit the similarity of $N'$ and $N''$.  This is
important because the current equivalence checkers exploiting such
similarity are \ti{incomplete}. If $N'$ and $N''$ are not ``similar
enough'', e.g. they have no functionally equivalent internal points,
the equivalence checker invokes a complete (but inefficient) procedure
ignoring similarities between $N'$ and $N''$.

Let $\mi{eq}(V',V'')$ specify a formula such that
\mbox{$\mi{eq}(\pnt{v}\,',\pnt{v}\,'')$ = 1} iff $\pnt{v}\,' =
\pnt{v}\,''$ where $\pnt{v}\,', \pnt{v}\,''$ are full assignments to
$V'$ and $V''$ respectively. Let formulas $G'(X',V',w')$ and
$G''(X'',V'',w'')$ specify $N'$ and $N''$ respectively. (As usual, we
assume that a formula $G$ specifying a circuit $N$ is obtained by
Tseitin transformations~\cite{tseitin}, see
Section~\ref{sec:prop_gen}.)
Let $h(w',w'')$ be a formula obtained by taking \ti{eq} out of
\prob{Z}{\mi{eq} \wedge G' \wedge G''} where $Z = X'\cup V' \cup X''
\cup V''$.  That is \prob{Z}{\mi{eq} \!\wedge\!G'\!\wedge\!G''}
$\equiv$ $h \wedge \prob{Z}{G' \wedge G''}$.  If $h \imp (w' \equiv
w'')$, then $N'$ and $N''$ are equivalent. Otherwise, $N'$ and $N''$
are inequivalent, unless they are identical constants
i.e. $w'\!\equiv\!w''\!\equiv 1$ or \mbox{$w'\!\equiv
  w''\!\equiv\!0$}. It is formally proved in~\cite{fmcad16} that the
more similar $N',N''$ are (where similarity is defined in the most
general sense), the easier taking $\mi{eq}$ out of \prob{Z}{\mi{eq}
  \wedge G' \wedge G''} becomes.

%
%
\vspace{4pt}
\noindent\tb{Model checking by PQE~\nf{\cite{mc_no_inv2}}.}  One can
use PQE to find the reachability diameter of a sequential circuit
without computing the set of all reachable states. So, one can prove
an invariant by PQE without generating a stronger invariant that is
inductive.
Let $T(S',V,S'')$ denote the transition relation of a sequential
circuit $N$ where $S',S''$ are the present and next state variables
and $V$ specifies the (combinational) input variables. Let $I(S)$
specify the initial states of $N$.  For the sake of simplicity, we
assume that $N$ can stutter i.e.  for every state \pnt{s} there exists
a full assignment \pnt{v} to $V$ such that
$T(\pnt{s},\pnt{v},\pnt{s})=1$, (Then the sets of states reachable in
$m$ transitions and \ti{at most} $m$ transitions are identical. If $T$
has no stuttering, it can be easily introduced by adding a variable to
$V$.)

Let \di{I,T} denote the \ti{reachability diameter} for initial states
$I$ and transition relation $T$.  That is every state of the circuit
$N$ can be reached in at most \di{I,T} transitions. Given a number
$m$, one can use PQE to decide if \mbox{$\di{I,T} < m$}. This is done
by checking if $I_1$ is redundant in \linebreak\prob{\Abs{m}}{I_0
  \wedge I_1 \wedge T_m}. Here $I_0$ and $I_1$ specify the initial
states of $N$ in terms of variables of $S_0$ and $S_1$ respectively,
$\Abs{m} = S_0 \cup V_0 \cup \dots \cup S_{m-1} \cup V_{m-1}$ and $T_m
= T(S_0,V_0,S_1) \wedge \dots \wedge T(S_{m-1},V_{m-1},S_m)$.  If
$I_1$ is redundant, then $\di{I,T} < m$.

The idea above can be used, for instance, to prove an invariant $P$
true in an IC3-like manner (i.e. by constraining $P$) but without
generating an inductive invariant. To prove $P$ true, it suffices to
constrain $P$ to a formula $H$ such that\linebreak a) $I \imp H \imp
P$, b) $\di{H,T} < m$ and c) no state falsifying $P$ can be reached
from a state satisfying $H$ in $m\!-\!1$ transitions. The conditions
b) and c) can be verified by PQE and bounded model checking~\cite{bmc}
respectively.  In the special case where $H$ meets the three
conditions above for $m=1$, it is an \ti{inductive invariant}.

\section{Deciding If A Property Is Unwanted}
\label{app:unw_props}
In this appendix, we give the main idea of a procedure for deciding if
a property is unwanted. We reuse the notation of
Section~\ref{sec:prop_gen}. That is, we consider property generation
for a combinational circuit $M(X,V,W)$ specified by formula
$F(X,V,W)$. Here $X,V,W$ denote the sets of the internal, input, and
output variables of $M$ respectively. For the sake of simplicity, we
consider generation of a property $H(W)$ depending only on output
variables of $M$. Such a property can be obtained by taking a clause
$C$ out of formula \Prob{X}{V}{F} (where only variables of $W$ are not
quantified).  So, $\Prob{X}{V}{F} \equiv H \wedge \Prob{X}{V}{F
  \setminus \s{C}}$.

Note that every clause of $H$ is a property too.  $H$ is an unwanted
property if and only if one of its clauses is an unwanted property.
Let $Q \in H$ be a single-clause unwanted property (and so $M$ has a
bug).  This property says that $M$ never outputs an assignment
\ti{falsifying} $Q$. To prove $Q$ unwanted one needs to find a test
\pqnt{v}{be} for which a correct version of $M$ would produce an
output falsifying $Q$. Here 'be' stands for 'bug-exposing'. (If $Q$ is
a \ti{desired} property, such \pqnt{v}{be} cannot be produced. So, the
failure to find \pqnt{v}{be} is an argument in favor of $Q$ being a
desired property.)

A test \pqnt{v}{be} can be found as follows. Since $F$ implies $Q$,
the formula $F \wedge \overline{Q}$ is unsatisfiable (and so $M$
cannot produce an output falsifying $Q$). Let $F'$ be a formula
obtained by removing some clauses from $F$ for which the formula $F'
\wedge \overline{Q}$ is satisfiable. Then means that the circuit $M'$
specified by $F'$ can produce an output falsifying $Q$. The intuition
here is that if $Q$ is unwanted property, the clauses removed from $F$
to obtain $F'$ relate to the buggy part of $M$. (Note that the circuit
$M'$ specified by $F'$ is non-deterministic. The reason is that $F'$ can
be satisfied by assignments (\ppnt{x}{1},\pnt{v},\ppnt{w}{1}) and
(\ppnt{x}{2},\pnt{v},\ppnt{w}{2}) with the same input \pnt{v} and
different outputs \ppnt{w}{1} and \ppnt{w}{2}).

Assume that \ti{all clauses} making up the buggy part of $M$ has been
removed from $F$ when obtaining $F'$. Then there is an assignment
(\pnt{x},\pqnt{v}{be},\pnt{w}) satisfying $F' \wedge
\overline{Q}$. That is a correct version of $M$ would produce the
assignment \pnt{w} falsifying $Q$ under the input \pqnt{v}{be}. So, to
obtain \pqnt{v}{be} one needs to build the formula $F'$ above and
examine the input part \pnt{v} of assignments satisfying $F' \wedge
\overline{Q}$.

The formula $F'$ can be obtained as follows. Recall that by our
assumption $F \setminus \s{C} \not\imp Q$ (see
Subsection~\ref{ssec:pg_gnrl_tsts} and Remark~\ref{rem:noise}).  So
the formula $F' \wedge \overline{Q}$ where $F' = F \setminus \s{C}$ is
satisfiable. However, one may not extract a test \pqnt{v}{be} from
assignments satisfying $F' \wedge \overline{Q}$ because $C$ may not be
the only clause making up the buggy part of $M$. A good heuristic for
forming $F'$ is to remove from $F$ every clause $B$ such that $F
\setminus \s{B} \not\imp Q$. (The latter means that one could obtain
the property $Q$ by taking out the clause $B$ from
\Prob{X}{V}{F}). However, the topic of obtaining formula $F'$ and
producing a test \pqnt{v}{be} needs further research.

\section{Combining PQE With Clause Splitting}
\label{app:cls_split}
In this appendix, we consider combining PQE with clause splitting
mentioned in Subsection~\ref{ssec:eff_comp}. We show that the
corresponding PQE problem of taking out a clause produced by splitting
is solved by \Eg in linear time.  We also show that if this clause is
not redundant, the solution produced by \Eg is a single-test
property.

Here we reuse the notation of Section~\ref{sec:prop_gen} but, for the
sake of simplicity, consider a single-output combinational circuit.
Let $M(X,V,w)$ be such a circuit where $X$ and $V$ specify the
internal and input variables respectively and $w$ is the output
variable of $M$. Let $F(X,V,w)$ be a formula specifying the circuit
$M$ and $C$ be a clause of $F$. Consider the case of
splitting \smallskip $C$ on \ti{all} variables of $V$. That is $C =
C_1 \wedge \dots \wedge C_{p+1}$ where $C_1 = C \vee
\overline{l(v_1)}$,\dots,\,$C_p = C \vee \overline{l(v_p)}$,\,$C_{p+1}
= C \vee l(v_1) \vee \dots \vee l(v_p)$ and $l(v_i)$ is a literal of
$v_i$ and $V = \s{v_1,\dots,v_p}$. Let $F'$ denote the formula
obtained from $F$ by replacing the clause $C$ with $C_1 \wedge \dots
\wedge C_{p+1}$. Denote by \bm{\pqnt{v}{spl}} the input assignment
falsifying the literals $l(v_1),\dots,l(v_p)$ where 'spl' stands for
'splitting'.

Consider applying \Eg to solve the PQE problem of taking the clause
$C_{p+1}$ out of \prob{X}{F'}. \smallskip \Eg starts with looking for
an assignment satisfying\linebreak $(F' \setminus \s{C_{p+1}}) \wedge
\overline{C_{p+1}}$ (to find a subspace where $F' \setminus
\s{C_{p+1}}$ does not imply $C_{p+1}$). Consider the following three
cases. The first case is that the formula above is unsatisfiable. Then
$C_{p+1}$ is trivially redundant in $F'$ and hence in \prob{X}{F'} and
\Eg terminates.

The second case is that there is an assignment
($\pnt{x},\pqnt{v}{spl},w^*$) satisfying $F'$ where \pnt{x} is a full
assignment to $X$ and $w^*$ is the output value taken by $M$ under the
input \pqnt{v}{spl}.  (Note \smallskip that any full assignment to $V$
that is different from \pqnt{v}{spl} falsifies
$\overline{C_{p+1}}$. So, any assignment satisfying $(F' \setminus
\s{C_{p+1}}) \wedge \overline{C_{p+1}}$ has to contain \pqnt{v}{spl}.)
Then formula $F'$ is satisfiable in subspace $(\pqnt{v}{spl},w^*)$ and
\Eg adds the plugging \smallskip clause $D(V,w)$ that is the longest
clause falsified by $(\pqnt{v}{spl},w^*)$. If $(F' \setminus
\s{C_{p+1}}) \wedge \overline{C_{p+1}} \wedge D$ is unsatisfiable,
then $C_{p+1}$ is redundant in \prob{X}{F'} and \Eg terminates.

The third case \smallskip occurs when there is an assignment
($\pnt{x},\pqnt{v}{spl},w^*$) satisfying\linebreak $(F' \setminus
\s{C_{p+1}}) \wedge \overline{C_{p+1}}$ where $w^*$ is the
\ti{negation} of the output value taken by $M$ under input
\pqnt{v}{spl}.  In this case, formula $F'$ is unsatisfiable in
subspace $(\pqnt{v}{spl},w^*)$. Since, $F' \setminus \s{C_{p+1}}$ is
satisfiable in this subspace, $C_{p+1}$ is \ti{not} redundant in
\prob{X}{F'}. To \ti{make} $C_{p+1}$ redundant in subspace
$(\pqnt{v}{spl},w^*)$, \Eg has to add the clause $B(V,w)$ that is the
longest clause falsified by $(\pqnt{v}{spl},w^*)$. The clause $B$ is a
solution to the PQE problem at hand i.e. $\prob{X}{F'} \equiv B \wedge
\prob{X}{F' \setminus \s{C_{p+1}}}$.

The clause $B$ above is implied by $F'$ (and hence $F$) and so, is a
\tb{property} of $M$. This property specifies the input/output
behavior of $M$ under the input \pqnt{v}{spl}. Namely, to satisfy $B$
when the variables of $V$ are assigned as in \pqnt{v}{spl}, one has to
set the variable $w$ to $\overline{w^*}$. The latter is the output
produced by $M$ under the input \pqnt{v}{spl}. So, the property $B$
specifies the behavior of $M$ under a \tb{single test}. In all three
cases above, the SAT problem considered by \Eg is solved just by
initial BCP. (The reason is that the formula at hand contains the unit
clauses produced by negating $C_{p+1}$ or those specifying the
subspace $(\pqnt{v}{spl},w^*)$.) So, \Eg solves the PQE problem above
in \tb{linear} time.

\section{Tests Specified By A Property Generated By PQE}
\label{app:tests_props}
In this appendix, we show the relation between tests specified by a
property obtained via PQE (see Subsection~\ref{ssec:tests_props}) and
those detecting stuck-at faults.  Here, we reuse the notation of
Section~\ref{sec:prop_gen}. Let $M(X,V,W)$ be a combinational circuit
where $X,V,W$ are the internal, input and output variables
respectively.  Let $F(X,V,W)$ be a formula specifying $M$.
Let $G$ be an AND gate of $M$ whose functionality is $x_3 = x_1 \wedge
x_2$. That is $x_1,x_2$ are the input variables of $G$ and $x_3$ is
its output variable.  The functionality of $G$ is specified by the
formula $C_1 \wedge C_2 \wedge C_3$ where $C_1 = \overline{x}_1 \vee
\overline{x}_2 \vee x_3$, $C_2 = x_1 \vee \overline{x}_3$, $C_3 = x_2
\vee \overline{x}_3$ (see Example~\ref{exmp:gate_cnf}). The clauses
$C_1,C_2,C_3$ are present in formula $F$. Consider taking $C_1$ out of
\prob{X}{F}. This clause makes $G$ produce the output value 1 when its
input values are 1.  (If $x_1$ and $x_2$ are set to 1, the clause
$C_1$ can be satisfied only by setting $x_3$ to 1.)

Let $H(V,W)$ be the property obtained by taking out $C_1$. That
is\linebreak $\prob{X}{F} \equiv H \wedge \prob{X}{F \setminus
  \s{C_1}}$. Let $Q(V,W)$ be a clause of $H$. As we mentioned earlier,
we assume that $H$ does not have redundant clauses i.e. those implied
by $F \setminus \s{C_1}$. Then the formula $(F \setminus \s{C_1})
\wedge \overline{Q}$ is satisfiable. Let (\pnt{x},\pnt{v},\pnt{w}) be
an assignment satisfying this formula. Note that this assignment
\ti{falsifies} $C_1$.  (Indeed, assume the contrary. Then
(\pnt{x},\pnt{v},\pnt{w}) satisfies $F$ because it already satisfies
$F \setminus \s{C_1}$. Since this assignment falsifies $Q$, we have to
conclude that $F \not\imp Q$ and hence $F \not\imp H$. So we have a
contradiction.)

The fact that (\pnt{x},\pnt{v},\pnt{w}) falsifies $C_1$ and satisfies
$F \setminus \s{C_1}$ means that one can view this assignment as an
execution trace of a faulty version \Sub{M}{flt} of $M$.  Namely, the
output $x_3$ of gate $G$ is stuck at 0 in \Sub{M}{flt}. (The clause
$C_1$ is falsified when $x_1=1,x_2=1,x_3=0$ i.e. if the gate $G$
outputs 0 when its input variables are assigned 1.)  Let
$(\pnt{x}^*,\pnt{v},\pnt{w}^*)$ be the execution trace of $M$ under
the input \pnt{v}. As we showed in Subsection~\ref{ssec:tests_props},
$\pnt{w}^*$ is different from \pnt{w}. So the input \pnt{v} exposes a
stuck-at fault by making \Sub{M}{flt} and $M$ produce different
outputs.

\section{Brief Description Of \dpqe}
\label{app:ds_pqe}
In this appendix, we give a high-level view of \dpqe and explain how
it works in \egp in more detail. A full description of \dpqe can be
found in~\cite{hvc-14}. \dpqe is based on the machinery of
D-sequents~\cite{fmcad13} ('\ti{DS}' in the name \dpqe stands for
'D-sequent'). Given a formula \prob{X}{F(X,Y)} and an assignment
\pnt{p} to $X \cup Y$, a D-sequent is a record \olds{X}{F}{p}{C}
stating that clause $C$ is redundant in \prob{X}{F} in subspace
\pnt{p}. In \egp, \dpqe is called in subspaces \pnt{y} where $F$ is
satisfiable. (Here \pnt{y} is a full assignment to $Y$.)  \dpqe
terminates upon deriving a D-sequent $(\prob{X}{F},\pnt{y}^*)
\rightarrow C$ where $\pnt{y}^* \subseteq \pnt{y}$. Such derivation
means that $C$ is proved redundant in \prob{X}{F} in subspace \pnt{y}.
Then the plugging clause $D$ falsified by $\pnt{y}^*$ is generated
where $\V{D} = \mi{Vars}(\pnt{y}^*)$.

\dpqe derives the D-sequent $(\prob{X}{F},\pnt{y}^*) \rightarrow C$
above by branching on variables of $X$. A variable is assigned a value
either by a decision or during Boolean Constraint Propagation (BCP). A
branch of the search tree goes on until an atomic D-sequent is derived
for $C$. This occurs when proving $C$ in the current subspace becomes
trivial. When backtracking, \dpqe merges D-sequents derived in
different branches using a resolution like operation called
\ti{join}. For instance, the join operation applied to D-sequents
$(\prob{X}{F},\pnt{p}') \rightarrow C$ where $\pnt{p}' =
(y_1=0,x_1=0)$ and $(\prob{X}{F},\pnt{p}'') \rightarrow C$ where
$\pnt{p}'' = (y_2=1,x_1=1)$ produces the D-sequent \olds{X}{F}{p}{C}
where \pnt{p}=$(y_1=0,y_2=1)$.

\dpqe has three situations where an atomic D-sequent is generated.
First, when $C$ is blocked in the current subspace and hence is
redundant there. Then an atomic D-sequent $(\prob{X}{F},\pnt{p})
\rightarrow C$ is generated where \pnt{p} consists of assignments that
made $C$ blocked in the current subspace. For instance, in
Example~\ref{exmp:eg_pqe+}, \dpqe would generate an atomic D-sequent
$(\prob{X}{F},(y_1\!=\!1)) \rightarrow C$.  Second, an atomic
D-sequent is generated when $C$ is satisfied by an assignment $w=b$
where $w \in X \cup Y$ and $b \in \s{0,1}$. (This can be a decision
assignment or an assignment derived from a clause during BCP.) Then an
atomic D-sequent\linebreak $(\prob{X}{F},(w=b)) \rightarrow C$ is
built.  Third, an atomic D-sequent is generated when a conflict occurs
and a conflict clause \Sub{C}{cnfl} falsified in the current subspace
is derived. Adding \Sub{C}{cnfl} to $F$ makes $C$ redundant in the
current subspace. So, an atomic D-sequent $(\prob{X}{F},\pnt{p})
\rightarrow C$ is generated where \pnt{p} is the shortest assignment
falsifying \Sub{C}{cnfl}.

\section{Experiments With HWMCC-13 Benchmarks}
\label{app:exper2}
In Section~\ref{sec:inv_gen_exper}, we described experiments with
multi-property benchmarks of the HWMCC-13 set~\cite{hwmcc13}. In this
appendix, we provide some additional information. Each benchmark
consists of a sequential circuit $N$ and invariants $P_0,\dots,P_m$
that are supposed to hold for $N$. We will refer to the invariant
\Sub{P}{agg} equal to $P_0 \wedge \dots \wedge P_m$ as the
\ti{aggregate invariant}.  We applied PQE to the generation of
invariants of $N$ that may be \ti{unwanted}.  Since every invariant
$P$ implied by \Sub{P}{agg} \ti{must} hold, the necessary condition
for $P$ to be unwanted is $\Sub{P}{agg} \not\imp P$.

Similarly to the experiment of Section~\ref{sec:fifo_exper}, we used the
formula $F_k = I(S_0) \wedge T(S_0,V_0,S_1)$ $\wedge \dots \wedge
T(S_{k-1},V_{k-1},S_k)$ to generate invariants. The number $k$ of time
frames was in the range of \mbox{$2\!\leq\!k\!
  \leq\!10$}. Specifically, we set $k$ to the largest value in this
range where $|F_k|$ did not exceed 500,000 clauses. We discarded the
benchmarks with $|F_2|\!>\!500,000$. We also dropped the smallest
benchmarks. So, in the experiments, we used 98 out of the 178
benchmarks of the set.

We describe three experiments. In every experiment, we generated
properties $H(S_k)$ by taking out a clause of \prob{\Abs{k}}{F_k}
where $\Abs{k}\!=\!S_0 \cup V_0 \cup \dots \cup S_{k-1} \cup
V_{k-\!1}$. Property\ $H$ is a \ti{local invariant} claiming that no
state falsifying $H$ can be reached in $k$ transitions. As in the
experiment of Section~\ref{sec:fifo_exper}, we took out only clauses
containing an unquantified variable (i.e a variable of $S_k$). In all
experiments, the \tb{time limit} for solving a PQE problem was set to
10 sec.

%
%
%
%
\subsection{Experiment 1}
\label{ssec:exper1}


The objective of the \ti{first experiment} was to demonstrate that
\egp could compute $H$ for realistic designs. We also showed in this
experiment that PQE could be much easier than QE (see
Section~\ref{sec:inv_gen_exper}) and that \egp outperforms \dpqe and
\Eg.  In this experiment, for each benchmark out of 98 mentioned above
we generated PQE problems of taking a clause out of
\prob{\Abs{k}}{F_k}. Some of them were trivially solved by
preprocessing. The latter eliminated the blocked clauses of $F_k$ that
were easy to identify and ran BCP launched due to the unit clauses
specifying the initial state. In all experiments, we \ti{discarded}
problems solved by preprocessing i.e. we considered only
\tb{non-trivial} PQE problems.

%
%
%
\vspace{5pt}
\begin{wraptable}{l}{1.6in}
\centering
\scriptsize
\captionsetup{justification=centering}
\caption{\small{PQE problems of the first group}}
\vspace{-5pt}
\begin{tabular}{|p{29pt}|p{25pt}|p{20pt}|p{20pt}|} \hline pqe &
 total & \multicolumn{2}{c|}{finished} \\ \cline{3-4}
 solver & probl. & num. &perc.\\ \hline
\ti{ds-pqe}     & 3,736      &   1,132     &~30\%     \\ \hline
\ti{eg-pqe}    & 3,736      &   3,296 &~88\%          \\ \hline
\ti{eg-pqe}$^+$ &  3,736     & \tb{3,329}  &~\tb{89\%}\\ \hline
\end{tabular}             
\label{tbl:only_rbp}
\end{wraptable}

Let $C$ be a clause taken out of \prob{\Abs{k}}{F_k}. We used \Eg to
partition the PQE problems we tried into two groups. \ti{The first
  group} consisted of problems for which we ran \Eg with the time
limit of 10 sec. and it never encountered a subspace \pnt{s_k} where
$F_k$ was satisfiable. Here \pnt{s_k} is a full assignment to
$S_k$. (Recall that the variables of $S_k$ are the only unquantified
variables of \prob{\Abs{k}}{F_k}.) So, in every subspace \pnt{s_k}
tried by \Eg, formula $F_k$ was either unsatisfiable or
$(F_k\!\setminus\!\s{C})\!  \imp\!C$.  \ti{The second group} consisted
of problems where \Eg encountered subspaces where $F_k$ was
satisfiable.  For either group we generated up to 40 problems per
benchmark. For some benchmarks, the total number of non-trivial
problems generated for the first or second group was under 40.  Many
PQE problems of either group had hundreds of thousands of variables.

%
%
%
\vspace{5pt}
\begin{wraptable}{l}{1.6in}
\centering
\scriptsize
\captionsetup{justification=centering}
\caption{\small{PQE problems of the second group}}
\vspace{-5pt}
\begin{tabular}{|p{29pt}|p{25pt}|p{20pt}|p{20pt}|} \hline pqe &
 total & \multicolumn{2}{c|}{finished} \\ \cline{3-4}
 solver & probl. & num. &perc.\\ \hline
\ti{ds-pqe}     & 3,094      &  464       &  15\%   \\ \hline
\ti{eg-pqe}    & 3,094      &   74         & 2\%     \\ \hline
\ti{eg-pqe}$^+$ &  3,094     & \tb{830}    & \tb{27\%}\\ \hline
\end{tabular}        
\label{tbl:ubp}
\end{wraptable}

The results for the first group are shown in
Table~\ref{tbl:only_rbp}. The first column gives the name of a PQE
solver. The second column shows the number of PQE problems in the
first group. The last two columns give the number and percentage of
problems finished in the time limit of 10 sec.
Table~\ref{tbl:only_rbp} shows that \Eg and \egp performed quite well
finishing a very high percentage of problems.  The results of \dpqe
are much poorer because it does not check if $(F_k \setminus \s{C})
\imp C$ in the current subspace i.e. if $C$ is trivially redundant.

The results for the second group are shown in Table~\ref{tbl:ubp} that
has the same structure as Table~\ref{tbl:only_rbp}. In particular, the
second column gives the number of PQE problems in the second
group. Table~\ref{tbl:ubp} shows that \Eg and \egp finished 2\% and
27\% of the problems respectively. So, \egp significantly outperformed
\Eg. The reason is that \Eg uses a satisfying assignment as a proof of
redundancy of the clause $C$ in subspace \ppnt{s}{k}.

%
\subsection{Experiment 2}
\label{ssec:exper2}
The second experiment was an extension of the first one. Its goal was
to show that PQE can generate invariants for realistic designs. For
each clause $Q$ of a local invariant $H$ generated by PQE we used \ict
to verify if $Q$ was a global invariant.  If so, we checked if
$\Sub{P}{agg} \not\imp Q$ held. To make the experiment less time
consuming, in addition to the time limit of 10 sec.  per PQE problem,
we imposed the following constraints.  First, we stopped a PQE-solver
even before the time limit if it generated more than 5 free
clauses. Second, the time limit for \ict was set to 30 sec.  Third,
instead of constraining the number of PQE problems per benchmark
(i.e. the number of single clauses taken out of \prob{\Abs{k}}{F_k})
like in the first experiment, we imposed the following two
constraints.  First, we stopped processing a benchmark as soon as the
total of 100 free clauses was generated (for all the PQE problems
generated for this benchmark). Second, we stopped processing a
benchmark even earlier if the summary run time for all PQE problems
generated for this benchmark exceeded 2,000 sec.

\begin{wraptable}{l}{3.2in}
\centering
\scriptsize
\captionsetup{justification=centering}
\caption{\small{A sample of HWMCC-13 benchmarks from\\ the experiment with \egp}}
\begin{tabular}{|p{20pt}|p{20pt}|p{20pt}|p{20pt}|p{20pt}|p{20pt}|p{17pt}|p{12pt}|p{15pt}|p{24pt}|} \hline
  name     & lat- & in-   &  time   &clau-   &\multicolumn{5}{c|}{single-clause properties} \\ \cline{6-10}
           & ches & var.      &  fra-   &ses     & gen.    & \multicolumn{3}{c|}{invariant?} & not   \\ \cline{7-9}
           &      & of   &  mes    &taken   & props   & un-  & no  &  yes & impl.   \\
           &      &   \Sub{P}{\!agg}      &         &out     &         & dec. &     &      & by \Sub{P}{\!agg}  \\ \hline

6s306  & 7,986 & ~25 & ~~2 & ~182  & ~100  & ~0 & ~94 & ~6 & ~6 \\ \hline    
6s176  & 1,566 & ~952 & ~~3 & ~31  & ~100  & ~23 & ~11 & ~66 & ~11  \\ \hline
6s428  &3,790  & ~340 & ~~4 & ~28  & ~100  & ~9  & ~5 & ~86 & ~83 \\ \hline
6s292  &3,190  & ~247 & ~~5 & ~24  & ~100  & ~41 & ~0  & ~59 & ~59 \\ \hline
6s156  & 513   & ~32  & ~~6 & ~113 & ~100  & ~0  & ~0  & ~100   & ~100   \\ \hline
6s275  & 3,196 & ~673 & ~~7 & ~20  & ~100  & ~0  & ~50 & ~50  & ~50 \\ \hline
6s325  &1,756  & ~301 & ~~8 & ~22  & ~100  & ~0  & ~1  & ~99  & ~97  \\ \hline
6s391  &2,686  & ~387 & ~~9 & ~35  & ~100  & ~1  & ~29  & ~70 & ~70 \\ \hline
6s282  &1,933  & ~20  & ~10 & ~111 & ~100  & ~0  & ~64 & ~36  & ~35 \\ \hline
\end{tabular}                
\label{tbl:sample}
\end{wraptable}

A sample of 9 benchmarks out of the 98 we used in the experiment with
\egp is shown in Table~\ref{tbl:sample}. Let us explain the structure
of this table by the benchmark 6s306 (the first line of the
table). The name of this benchmark is shown in the first column. The
second column gives the number of latches (7,986). The number of
invariants that should hold for 6s306 is provided in the third column
(25). So, the aggregate invariant \Sub{P}{\!agg} of 6s306 is the
conjunction of those 25 invariants. The fourth column shows that the
number $k$ of time frames for 6s306 was set to 2 (since
$|F_3|\!>\!500,000$).  The value 182 shown in the fifth column is the
total number of single clauses taken out of \prob{\Abs{k}}{F_k}
i.e. the number of PQE problems (where $k=2$ for 6s306).

Every free clause $Q$ generated by \egp when taking a clause out of
\prob{\Abs{k}}{F_k} was stored as a local single-clause invariant.
The sixth column shows that taking clauses out of the scope of
quantifiers was terminated when 100 free clauses (specifying 100 local
single-clause invariants) were generated.  Each of these 100 local
invariants held in $k$-th time frame. The following three columns show
how many of those 100 local invariants were true globally. \ict
finished every problem out of 100 in the time limit of 30 sec. So, the
number of undecided invariants was 0. The number of invariants \ict
proved false or true globally was 94 and 6 respectively. The last
column gives the number of global invariants \ti{not} implied by
\Sub{P}{\!agg} i.e. invariants that may be unwanted. For 6s306, this
number is 6.



\vspace{-7pt}
\subsection{Experiment 3}
\label{ssec:exper3}
To prove an invariant $P$ true, \ict conjoins it with clauses
$Q_1,\dots,Q_n$ to make $P\wedge Q_1\!\wedge \dots \wedge Q_n$
inductive.  If \ict succeeds, every $Q_i$ is an invariant. Moreover,
$Q_i$ may be an \ti{unwanted} invariant.  Arguably, the cause of
efficiency of \ict is that $P$ is often close to an inductive
invariant. So, \ict needs to generate a relatively small number of
clauses $Q_i$ to make the constrained version of $P$
inductive. However, this nice feature of \ict drastically limits the
set of invariant clauses it generates. In this subsection, we
substantiate this claim by an experiment. In this experiment, we
picked the HWMCC-13 benchmarks for which one could prove \ti{all}
predefined invariants $P_1, \dots, P_m$ within a time limit. Namely,
for every benchmark we formed the aggregate invariant \mbox{\spe =
  $P_1 \wedge \dots \wedge P_m$} and ran \ict to prove \spe true.

We selected the benchmarks that \ict solved in less than 1,000 sec.
(In addition to dropping the benchmarks not solved in 1,000 sec., we
discarded those where \spe failed because some invariants $P_i$ were
false). Let \bm{\sspe} denote \smallskip the inductive version of \spe
produced by \ict when proving \spe true.  That is, $\sspe$ is \spe
conjoined with the invariant clauses produced by \ict.

\begin{wraptable}{l}{2.4in}
\centering
\scriptsize
\captionsetup{justification=centering}
\caption{\small{Invariants of \egp and \\ \ict}}
\begin{tabular}{|p{20pt}|p{20pt}|p{22pt}|p{20pt}|p{25pt}|p{25pt}|} \hline
  name     & lat- & inva-  &\multicolumn{3}{c|}{glob. single cls. invars} \\ \cline{4-6}
           & ches &riants    & glob.    & not & not   \\
           &      & of    & inva-   & impl.   & impl.  \\
           &      & \spe &  riants & by $\spe$  &by \sspe  \\ \hline
  6s135  & 2,307  &~340 &~53 &~53    &~27 \\ \hline
  6s325  & 1,756  &~301 &~99 &~99   &~96\\ \hline
ex1    & 130    &~33  &~25 &~16    &~16 \\ \hline
ex2    & 212    &~32  &~64  &~64    &~47 \\ \hline
6s106  & 135    &~17  &~100   &~96    &~96 \\ \hline
6s256  & 3,141  &~5   &~6  &~6    &~6  \\ \hline
ex3    & 61     &~3   &~4  &~4    &~4  \\ \hline
ex4    & 63     &~3   &~1  &~1    &~1  \\ \hline
6s209  & 5,759  &~2   &~95 &~95   &~89 \\ \hline
6s113  & 994    &~1   &~19 &~16   &~16 \\ \hline
6s143  & 260    &~1   &~97  &~86   &~77  \\ \hline
6s170  & 3,141  &~1   &~13  &~13    &~13  \\ \hline
6s252  & 170   &~1   &~55  &~41   &~34 \\ \hline \hline
\tb{Total}  &       &     &     &  \tb{590}  & \tb{522} \\ \hline
\end{tabular}
\label{tbl:ic3_pqe}
\end{wraptable}

For each of the selected benchmarks we generated invariants by \egp
exactly as in Experiment 2. That is, we stopped generation of local
single clause invariants when their number exceeded 100. Then we ran
\ict to identify local invariants that were global as well.  After
that we checked which of the global invariants generated by \egp were
not implied by \spe. The difference from Experiment 2 was that we also
checked which global invariants generated by \egp were not implied by
$\sspe$.

The results of the experiment are shown in Table~\ref{tbl:ic3_pqe}.
The first three columns of this table are the same as in
Table~\ref{tbl:sample}. They give the name of a benchmark, the number
of latches and the number of invariants $P_1$,$\dots$,$P_m$ to
prove. (The actual names of examples \ti{ex1},..,\ti{ex4} in the
HWMCC-13 set are \ti{pdtvsarmultip}, \ti{bobtuintmulti},
\ti{nusmvdme1d3multi}, \ti{nusmvdme2d3multi} respectively.) The next
column of Table~\ref{tbl:ic3_pqe} shows the number of local invariants
generated by \egp that turn out to be global.  The last two columns
give the number of global invariants that were not implied by \spe and
\sspe respectively. The last row of the table shows that in 522 cases
out of 590 the invariants not implied by \spe were not implied by
\sspe either. So, in 88\% of cases, the invariant clauses generated by
\egp were \ti{different} from those generated by \ict to form \sspe.

\vspace{-5pt}
\section{Experiment With Combinational Circuits}
\label{app:symb_sim}
In this appendix, we give more information about the experiment with
property generation for combinational circuits described in
Section~\ref{sec:comb_exper}. We formed PQE problems as follows.  For
each benchmark $N$ we picked the number $k$ of time frames to unroll.
The value of $k$ ranged from 10 to 40. (For larger circuits we picked
a smaller value of $k$.) Then we unrolled $N$ for $k$ time frames to
form a combinational circuit $M_k$ and randomly generated a clause
$B(S_k)$ of 15 literals. So, $B$ depended on output variables of
$M_k$. After that, we constructed the subcircuit $M'_k$ of $M_k$ as
described in Section~\ref{sec:comb_exper}. That is, $M'_k$ was
obtained by removing the logic of $M_k$ that did not feed any output
variable present in $B$.

Let formula $F'_k$ specify the subcircuit $M'_k$. (Here we reuse the
notation of Section~\ref{sec:comb_exper}.)  For every benchmark, we
generated PQE problems of taking different clauses $C$ out of
\prob{S_{1,k}}{F'_k}. That is each PQE problem was to find $H$ such
that $\prob{S_{1,k}}{F'_k} \equiv H \wedge \prob{S_{1,k}}{F'_k
  \setminus \s{C}}$. Each clause $C$ to take out was chosen among the
clauses of $F'_k$ that contained a variable of $S_k$ (i.e. an output
variable of $M'_k$). In this way, we formed a set of 3,254 PQE
problems.  1,668 of these problems were solved by simple formula
preprocessing i.e. turned out to be trivial.  So, in the experiment we
used the remaining 1,586 non-trivial PQE problems.

%
%
\vspace{5pt}
\begin{wraptable}{l}{2.3in}
\centering
\scriptsize
\captionsetup{justification=centering}
\caption{\small{Property generation}}
\vspace{-5pt}
  \begin{tabular}{|p{29pt}|p{22pt}|p{20pt}|p{20pt}|p{20pt}|p{20pt}|} \hline
 pqe & num.  & \multicolumn{4}{c|}{properties were generated}   \\ \cline{3-6}
 solver         & of pqe  &num-   &   &  \multicolumn{2}{c|}{stronger than}     \\
                &prob-    &ber     &~~\%      & \multicolumn{2}{c|}{3-val. sim.} \\ \cline{5-6}
                &lems       &   &   & num. & ~~\%  \\ \hline
\ti{ds-pqe}     & 1,586 & 983 &~~62 &~728  &~~74    \\ \hline
\ti{eg-pqe}     & 1,586 &  450 &~~28 &~361 &~~\tb{80}   \\ \hline
\ti{eg-pqe}$^+$ & 1,586 & \tb{1,046}  &~~\tb{66} &~\tb{817} &~~78   \\ \hline 
\end{tabular}                
\label{tbl:all_prop_gen}
\end{wraptable}

The time limit for solving a PQE problem was set to 10 sec. Besides,
solving a PQE problem terminated as soon as the size of $H$ reached 5
clauses.  The results of the experiment are summarized in
Table~\ref{tbl:all_prop_gen}. The second column gives the total number
of PQE problems.  The next two columns show the number and percentage
of problems where $H$ was non-empty i.e. had at least one
clause. (Recall, that each clause of $H$ represents a property.) The
last two columns give the number and percentage of cases where a
clause of $H$ represented a property that was stronger than ones
produced by 3-valued simulation, a version of symbolic
simulation~\cite{SymbolSim}.  Consider, for instance, the last line of
the table corresponding to \egp.  For 817 out of 1,046 PQE problems
where $H$ was not empty, at least one clause of $H$ constituted a
property that could not be produced by 3-valued
simulation. Table~\ref{tbl:all_prop_gen} shows that \Eg had the
weakest results generating properties only for 28\% of problems
whereas \dpqe and \egp performed much better producing properties for
62\% and 66\% problems respectively.

\end{document}